\documentclass[12pt,authoryear,review]{elsarticle}

\makeatletter
\def\ps@pprintTitle{%
 \let\@oddhead\@empty
 \let\@evenhead\@empty
 \def\@oddfoot{}%
 \let\@evenfoot\@oddfoot}
\makeatother

\usepackage{color}
\usepackage{booktabs}
\usepackage{mathtools}
\usepackage{amsmath,amssymb,amsfonts}
\usepackage{mathrsfs}
\usepackage{graphics}
\usepackage{graphicx}
\usepackage{dsfont}
\usepackage{subfigure}
\usepackage{url}
\usepackage[labelfont=bf,font={footnotesize,it},format=hang,width=0.8\textwidth]{caption}
\usepackage[titletoc,title]{appendix}
\usepackage{enumitem}
\usepackage{rotating}
\usepackage{pdflscape}
\usepackage{afterpage}

\setitemize{itemsep=0pt}
\numberwithin{equation}{section}
\newtheorem{Theorem}{Theorem}
\newtheorem{Lemma}{Lemma}
\newtheorem{Proposition}{Proposition}
\newtheorem{Corollary}{Corollary}
\newtheorem{Remark}{Remark}
\newtheorem{Assumption}{Assumption}
\newtheorem{Definition}{Definition}
\newenvironment{proof}[1][Proof]{\textbf{#1.} }{\ \rule{0.5em}{0.5em}}

\newcommand{\EE}{{\mathbb E}}
\newcommand{\NN}{{\mathbb N}}

\newcommand{\RR}{{\mathbb R}}
\newcommand{\PP}{{\mathbb P}}
\newcommand{\LL}{{\mathbb L}}
\newcommand{\ind}{\mathds{1}}
\newcommand{\As}{{\mathcal A}}

\newcommand{\Ms}{{\mathcal M}}

\newcommand{\Ft}{{\mathcal F}}
\newcommand{\FF}{{\mathfrak F}}
\newcommand{\Gt}{{\mathcal G}}
\newcommand{\GG}{{\mathfrak G}}

\newcommand{\ones}{\mathbf{1}}
\newcommand{\vB}{{\boldsymbol B}}
\newcommand{\vX}{{\boldsymbol X}}
\newcommand{\vx}{{\boldsymbol x}}
\newcommand{\vW}{{\boldsymbol W}}
\newcommand{\vM}{{\boldsymbol M}}

\newcommand{\vp}{{\boldsymbol p}}
\newcommand{\vP}{{\boldsymbol{P}}}

\newcommand{\vxi}{{\boldsymbol \xi}}
\newcommand{\vpi}{{\boldsymbol \pi}}

\newcommand{\vphi}{{\boldsymbol \phi}}
\newcommand{\vPhi}{{\boldsymbol \Phi}}
\newcommand{\veta}{{\boldsymbol \eta}}
\newcommand{\vrho}{{\boldsymbol \rho}}
\newcommand{\vzeta}{{\boldsymbol \zeta}}
\newcommand{\valpha}{{\boldsymbol \alpha}}
\newcommand{\vgamma}{{\boldsymbol \gamma}}

\newcommand{\vlambda}{{\boldsymbol{\lambda}}}

\newcommand{\vMHat}{\widehat{\vM}}
\newcommand{\MHat}{\widehat{M}}
\newcommand{\vWtilde}{\widetilde{\vW}}
\newcommand{\gammaHat}{\widehat{\gamma}}
\newcommand{\valphaHat}{\widehat{\boldsymbol \alpha}}
\newcommand{\vgammaHat}{\widehat{\boldsymbol \gamma}}

\newcommand{\mI}{{\mathbf{I}}}
\newcommand{\mG}{{\mathbf{G}}}
\newcommand{\mQ}{{\mathbf{Q}}}
\newcommand{\mB}{{\mathbf{B}}}
\newcommand{\mSigma}{\mathbf \Sigma}
\newcommand{\mOmega}{\mathbf \Omega}
\newcommand{\mA}{{\mathbf A}}
\newcommand{\mZ}{{\mathbf Z}}
\newcommand{\mSigmaHat}{\widehat{\mathbf \Sigma}}
\newcommand{\mx}{{\mathbf{x}}}

\newcommand{\mfN}{{\mathfrak{N}}}
\newcommand{\mfM}{{\mathfrak{M}}}

\newcommand{\mDim}[3] { \underset{\color{red} \footnotesize #2 \times #3}{#1} }

\newcommand\redunderbrace[2]{{\color{red} \underbrace{\color{black} #1}_{\tiny \text{#2}} }}

\usepackage{array}
\newcolumntype{L}[1]{>{\raggedright\let\newline\\\arraybackslash\hspace{0pt}}m{#1}}
\newcolumntype{C}[1]{>{\centering\let\newline\\\arraybackslash\hspace{0pt}}m{#1}}
\newcolumntype{R}[1]{>{\raggedleft\let\newline\\\arraybackslash\hspace{0pt}}m{#1}}

\begin{document}


\begin{frontmatter}

\setlength{\parskip}{0em}

\title{\textbf{Active and Passive Portfolio Management with Latent Factors}
\\[1em]}
\tnotetext[t1]{The authors would like to thank NSERC for partially funding this work.  }

\author{Ali Al-Aradi}
\ead{ali.al.aradi@utoronto.ca}

\author{Sebastian Jaimungal}
\ead{sebastian.jaimungal@utoronto.ca}
\address{Department of Statistical Sciences, University of Toronto}

\begin{abstract}
We address a portfolio selection problem that combines active (outperformance) and passive (tracking) objectives using techniques from convex analysis. We assume a general semimartingale market model where the assets' growth rate processes are driven by a latent factor. Using techniques from convex analysis we obtain a closed-form solution for the optimal portfolio and provide a theorem establishing its uniqueness. The motivation for incorporating latent factors is to achieve improved growth rate estimation, an otherwise notoriously difficult task. To this end, we focus on a model where growth rates are driven by an unobservable Markov chain. The solution in this case requires a filtering step to obtain posterior probabilities for the state of the Markov chain from asset price information, which are subsequently used to find the optimal allocation. We show the optimal strategy is the posterior average of the optimal strategies the investor would have held in each state assuming the Markov chain remains in that state. Finally, we implement a number of historical backtests to demonstrate the performance of the optimal portfolio. \\
\end{abstract}

\begin{keyword}
Active portfolio management;
Convex analysis;
Stochastic Portfolio Theory;
Functionally generated portfolios;
Rank-based models;	
Growth optimal portfolio;
Hidden Markov models;
Partial information.
\end{keyword}

\end{frontmatter}

\section{Introduction}




Problems in portfolio management can be divided into two types: active and passive. In the former, investors aim to achieve superior portfolio returns; in the latter, the investors' goal is to track a preselected index; see, for example, \cite{buckley1998} or \cite{pliska2004}. One can further separate active portfolio management objectives into two types: absolute and relative. There is a great deal of literature dedicated to solving various portfolio selection problems with absolute goals via stochastic control theory. The seminal work of \cite{Merton69} introduced the dynamic asset allocation and consumption problem, utilizing stochastic control techniques to derive optimal investment and consumption policies. Extensions can be found in \cite{Merton71}, \cite{Magill76}, \cite{Davis90}, \cite{Browne97} and more recently \cite{BlanchetScaillet2008}, \cite{Liu2013} and \cite{Ang2014}  to name a few. The focus in these papers is generally on maximizing the utility of discounted consumption and terminal wealth or minimizing shortfall probability, or other related \textit{absolute} performance measures that are independent of any external benchmark or relative goal. Works on optimal active portfolio management with relative goals (i.e. attempting to outperform a given benchmark) can be found in \cite{Browne1999a}, \cite{Browne1999b} and \cite{Browne2000}, \cite{pham2003large} and, more recently, \cite{oderda2015}.

There are also several works that address the question of achieving absolute portfolio selection goals when parameters are stochastic, including cases where the investor only has access to partial information and must rely on Bayesian learning or filtering techniques to solve for their optimal allocation. \cite{Merton71} solves for the portfolio that maximizes expected terminal wealth assuming that the instantaneous expected rate of return follows a mean-reverting diffusive process. \cite{Lakner98} extends this to the case where the drift processes are unobservable. In \cite{rieder2005} the assets' drift switches between various quantities according to an unobservable Markov chain; \cite{frey2012} extends this to incorporate expert opinions, in the form of signals at random discrete times, into the filtering problem by using this observable information to obtain posterior probabilities for the state of the Markov chain. \cite{bauerle2007} introduces jumps to the asset price dynamics by including Poisson random measures with unobservable intensity processes. Latent models are also central to the work of \cite{casgrain2016} and \cite{casgrain2018} in the context of algorithmic trading and mean field games.


Many of the concepts discussed in this work, particularly the notion of functionally generated portfolios (FGPs) and rank-based models, are key concepts in Stochastic Portfolio Theory (SPT) (see \cite{fernholz2002} and \cite{karatzas2009} for a thorough overview). SPT is a flexible framework for analyzing portfolio behavior and market structure which takes a descriptive, rather than a normative, approach to addressing these issues, and emphasizes the use of observable quantities to make its predictions and conclusions. The appeal of SPT partially lies in the fact that it relies on a minimal set of assumptions that are readily satisfied in real equity markets and that the techniques it employs construct relative arbitrage portfolios that outperform the market almost surely without the need for parameter estimation. This is primarily done through the machinery of portfolio generating functions (PGFs), which are portfolio maps that give rise to investing strategies that depend only on prevailing market weights. A discussion of the relative arbitrage properties of FGPs and related approaches to achieving outperformance vis-\`{a}-vis the market portfolio can be found in \cite{pal2013energy}, \cite{wong2015optimization} and \cite{pal2016geometry}.

Although SPT focuses on almost sure outperformance, i.e. relative arbitrage with respect to the benchmark portfolio, we deviate from this criterion in favor of maximizing the expected growth rate differential. We present two justifications for this choice. First, certain rank-based models such as the first-order models admit equivalent martingale measures over all horizons implying the non-existence of relative arbitrage opportunities. This forces the investor to select an alternative performance criterion. Secondly, \cite{fernholz2002} argues for the use of functionally-generated portfolios such as diversity-weighted portfolios as benchmarks for active equity portfolio management given their passive, rule-based nature and ease of implementation. However, \cite{wong2015optimization} notes that under certain reasonable conditions relative arbitrage opportunities do not exist with respect to these portfolios. Therefore, once again the investor must seek a substitute for almost sure outperformance if they decide to have a performance benchmark of this sort. One SPT-inspired work that uses an expectation-based objective function is \cite{vervuurt2016}, in which machine learning techniques are utilized to achieve outperformance in expectation by maximizing the investor's Sharpe ratio.

Active managers often dynamically invest in markets with the goal of achieving optimal relative returns against a \textit{performance benchmark} while anchoring their portfolio to a \textit{tracking benchmark} (in the sense of incurring minimal active risk/tracking error). They also often have in mind additional constraints on the investor's portfolio, e.g. penalizing large positions in certain assets or excessive volatility in the investor's wealth. In \cite{alaradi2018}, the authors formulate these goals and constraints by posing a portfolio optimization problem with log-utility of relative wealth, together with running penalty terms that incorporate the investor's constraints on tracking a benchmark and total risk. They solve the problem in closed-form using dynamic programming under the assumptions that the benchmarks are differentiable maps that are Markovian in the asset values; this encompasses the market portfolio and, more broadly, the class of (time-dependent) functionally generated portfolios.

A shortcoming of \cite{alaradi2018} is that when the investor values outperformance, the optimal solution relies crucially on the asset growth rate estimates, which are assumed to be bounded, differentiable, deterministic functions of time. However, returns are notoriously difficult to estimate robustly and the deterministic assumption does not provide adequate estimates.  To address this shortcoming, here, we allow for growth rates to be stochastic and be driven by latent factors. This is essential to making the strategy robust to differing market environments. Our formulation also accommodates rank-based models; such models exploit the stability of capital distribution in the market to arrive at estimates of asset growth rates based on asset ranks.

Our modeling assumption is similar to that adopted in \cite{casgrain2018}, who study the mean-field version of an algorithmic trading problem, where assets are driven by two components: a drift term and a martingale component both of which are adapted to an unobservable filtration. The investor's strategy, on the other hand, is restricted to be adapted to a smaller filtration; namely, the natural filtration generated by the price process.

The approach we take to solve the stochastic control problem is based on techniques from convex analysis as in \cite{bank2017hedging} and \cite{casgrain2018}, however these techniques date as far back as \cite{cvitanic1992convex}. The reason we deviate from the dynamic programming approach taken in \cite{alaradi2018} centers around the difficulty of extending that approach to more general market models. Although possible, it would be a  difficult task to ensure that all the additional state variables (which would include various semimartingale local times in the case of rank-based models) satisfy the conditions for a Feynman-Kac representation to the HJB equation that arises from the control problem, which is a central aspect of the proof. A number of additional (possibly restrictive) assumptions would have to be made on the market model and, as such, the approach we adopt in the current work allows for a more succinct solution to a more general problem with fewer assumptions. \\





\section{Model Setup} \label{sec:modelSetup}

\subsection{Market Model}

We adopt a market model that generalizes the one in \cite{alaradi2018} and is a multidimensional version of the one used in \cite{casgrain2018}. Let $(\Omega, \Gt, \GG, \PP)$ be a filtered probability space, where $\GG = \{\Gt_t\}_{t \geq 0}$ is the natural filtration generated by all processes in the model. We  assume that the market consists of $n$ assets defined as follows:
\begin{Definition}
The \textbf{stock price process} for asset $i$, $X^i = \left( X^i_t \right)_{t \geq 0}$ for all $i\in\mfN \coloneqq \{1,\dots,n\}$, is a positive semimartingale satisfying:
\begin{equation}
X^i_t = X^i_0 \exp \left( \int_0^t \gamma^i_s~ds + M^i_t \right)
\end{equation}
where $\gamma^i = \left( \gamma^i_t \right)_{t \geq 0}$ is a $\GG$-adapted process representing the asset's \textbf{(total) growth rate} and $M^i = \left( M^i_t \right)_{t \geq 0}$ is a $\GG$-adapted martingale with $M^i_0 = 0$ representing the asset's \textbf{noise component}.
\end{Definition}

It is convenient to work with the logarithmic representation of asset dynamics:
\begin{Proposition}
The logarithm of prices, $\log X^i$, satisfies the stochastic differential equation:
\begin{equation} \label{eqn:logPrice1D}
d \log X^i_t = \gamma^i_t ~dt + dM^i_t\,,\qquad \forall\;i\in\mfN\,.
\end{equation}
This can also be expressed in vector notation as follows:
\begin{equation} \label{eqn:logPrice}
d \log \vX_t = \vgamma_t ~dt + d\vM_t\,,
\end{equation}
where
{\small
\begin{equation*}
\mDim{\log \vX_t}{n}{1} = \left(\log X^1_t, ... , \log X^n_t \right)^\intercal\,,
\qquad
\mDim{\vgamma_t}{n}{1} = \left(\gamma^1_t, ... , \gamma^n_t \right)^\intercal\,,
\qquad
\mDim{\vM_t}{n}{1} = \left(M^1_t, ... , M^n_t\right)^\intercal\,.
\end{equation*}
}
\end{Proposition}

We make the following assumption on the growth rate and noise component of asset prices:

\begin{Assumption} \label{asmp:growthMtg}
The growth rate and martingale noise processes satisfy one of the two following conditions:
\begin{enumerate}[label=(\alph*),labelsep=0.5cm]
\item $\vgamma$, $\vM \in \LL^2$;
\item $\vgamma \in \LL^{\infty,M}$ and $\vM \in \LL^1$,
\end{enumerate}
\begin{align*}
\text{where } \qquad \LL^p &= \left\{ f: \Omega \times [0,T] \rightarrow \RR^n \text{ s.t. }  \EE \left[ \int_{0}^{T} \left(\| f_t \|_p\right)^p ~dt \right] < \infty \right\}\,, \qquad 0<p<\infty
\\
\LL^{\infty,M} &= \left\{ f: \Omega \times [0,T] \rightarrow \RR^n \text{ s.t. }  \sup_{t\in[0,T]}\| f_t \|_\infty \leq M, \; \PP-a.s. \right\}\,.
\end{align*}
\end{Assumption}
In the assumption above $\| \vx \|_p \coloneqq \left(\sum_{i=1}^n |x_i|^p \right)^{1/p}$ and $\| \vx \|_\infty \coloneqq \underset{i \in \mfN}{\max} \hspace{0.1cm} |x_i|$ for $\vx \in \RR^n$  denote the $p$-norm and $\infty$-norm on $\RR^n$, respectively. Furthermore, we will make use of the shorthand notation $\| \vx \| \coloneqq \| \vx \|_2$ to denote the usual Euclidean norm. 

We also assume the quadratic co-variation processes associated with the noise component satisfy
\begin{Assumption} \label{asmp:posDef}
Let $\mSigma$ be the matrix whose $ij$-th element is the \textbf{quadratic covariation} process between $M_i$ and $M_j$, $\Sigma_{ij} \coloneqq \langle M_i, M_j \rangle_t$. We assume that, for each $\mx \in \RR^n$, there exists $\varepsilon > 0$ and $C < \infty$    such that
\begin{align}
\varepsilon \| \mx \|^2 ~\leq~ \mx^\intercal \mSigma_t \mx  ~\leq~ C \| \mathbf{x} \|^2\,, \qquad
\forall t\ge0.
\end{align}
\end{Assumption}
This is an extension of the usual \textbf{non-degeneracy} and \textbf{bounded variance} conditions.

\begin{Remark}
The constant $M$ in $\LL^{\infty,M}$ of Assumption \ref{asmp:growthMtg} may depend on the constants $\varepsilon$ and $C$ that appear in Assumption \ref{asmp:posDef}, but provided that $M$ is sufficiently large we can ensure that the candidate optimal solution that we obtain is in fact in the set of admissible controls.
\end{Remark}

\subsection{Portfolios and Observable Information}

The investor does not have access to the latent processes driving asset prices and observes asset prices alone (it is possible to allow other processes in addition to the price process, but here we restrict to this case). Let the filtration $\FF = \left\{ \Ft_t \right\} _{t \geq 0}$ where $\Ft_t = \sigma \left( \left\{ \vX_s \right\}_{s\in[0,t]} \right)$ denotes the investor's filtration.

\begin{Definition}
A \textbf{portfolio} is a measurable, $\FF$-adapted, vector-valued process $\vpi = ( \vpi_t )_{t \geq 0}$, where $\vpi_t = \left(\pi^1_t, ..., \pi^n_t \right)^\intercal$ such that for all $t \geq 0$, $\vpi_t$ satisfies:
\begin{equation}
\pi^1_t + \cdots + \pi^n_t = 1 \quad \mbox{a.s.}
\end{equation}
Furthermore, we define the \textbf{set of admissible portfolios} as follows:
\begin{enumerate}[label=(\alph*),labelsep=0.5cm]
\item if Assumption \ref{asmp:growthMtg}(a) is enforced, we assume $\vpi \in \LL^2$ and define :
\begin{equation}
\As^2 = \big\{ \vpi: \Omega \times [0,T] \rightarrow \RR^n \text{ s.t. } \vpi \in \LL^2, ~  \FF\text{-adapted} \text{ and } \vpi_t^\intercal\ones = 1, \text{ for } t \geq 0 ~~\PP\text{-a.s.} \big\}
\end{equation}
\item if Assumption \ref{asmp:growthMtg}(b) is enforced, we assume $\vpi\in\LL^{\infty,M}$ and define:
\begin{equation}
\As^\infty = \big\{ \vpi: \Omega \times [0,T] \rightarrow \RR^n \text{ s.t. } \vpi\in\LL^{\infty,M},~  \FF\text{-adapted} \text{ and } \vpi_t^\intercal\ones = 1, \text{ for } t \geq 0 ~~\PP\text{-a.s.} \big\}
\end{equation}
\end{enumerate}
\end{Definition}
In the sequel, we write $\As^c$ to denote either $\As^2$ or $\As^\infty$ depending on which part of Assumption \ref{asmp:growthMtg} is being enforced.


\begin{Remark}
The cost of allowing for $\LL^1$ noise processes is that both growth rate processes and admissible portfolios are $\LL^{\infty,M}$ rather than $\LL^2$ processes.
\end{Remark}

In the definition above, portfolios are adapted to the filtration $\FF\subseteq\GG$, which is the information set generated by the asset price paths and not the full information set $\GG$. The latter includes the noise component $\vM$ as well as its quadratic covariation process $\mSigma$, both assumed unobservable. This ensures that strategies depend only on fully observable quantities which in our context are limited to the asset price processes.

Given the model dynamics, and portfolio assumptions, we next derive the dynamics of wealth associated with an arbitrary portfolio $\vpi$:
\begin{Proposition} \label{prop:portfolioDynamics}
The logarithm of the \textbf{portfolio value process} $Z^\vpi = (Z^\vpi_t )_{t \geq 0}$ satisfies the SDE:
\begin{equation} \label{eqn:portSDE}
d \log Z^\vpi_t = \gamma^\vpi_t ~dt + \vpi_t^\intercal ~ d\vM_t\,,
\end{equation}
\begin{align*}
where \qquad \gamma^\vpi_t &= \vpi_t^\intercal \vgamma_t + \Gamma^{\vpi}_t\,,
&& \Gamma^{\vpi}_t  = \tfrac{1}{2} \left[ \vpi_t^\intercal \mbox{diag}(\mSigma_t) - \vpi_t^\intercal \mSigma_t \vpi_t  \right]\,,
\end{align*}
and $\gamma^\vpi$ and $\Gamma^\vpi$ are the \textbf{portfolio growth rate} and \textbf{excess growth rate} processes, respectively. 
\end{Proposition}
\begin{proof}
The proof follows the same steps as the proof of Proposition 1.1.5. of \cite{fernholz2002}. The proportional change in the value of portfolio $\vpi$ is a weighted average of the simple return of each asset held in the portfolio:
	\[ \frac{dZ^\vpi_t}{Z^\vpi_t} = \sum_{i=1}^{n} \pi^i_t ~ \frac{d X^i_t}{X^i_t}\,. \]
	From the asset dynamics in \eqref{eqn:logPrice1D} and an application of It\^{o}'s lemma we have:
	\[ \frac{d X^i_t}{X^i_t} = \left(\gamma^i_t + \tfrac{1}{2} \langle M^i \rangle_t \right) dt + dM^i_t\,,  \]
	where $\langle M^i \rangle_t = \langle M^i, M^i \rangle_t $ is the quadratic variation process of $\log X^i$. Another application of It\^{o}'s lemma on the portfolio wealth process dynamics gives
	\[  \frac{dZ^\vpi_t}{Z^\vpi_t} = d \log Z^\vpi_t + \tfrac{1}{2} ~ d \langle  \log Z^\vpi \rangle_t\,. \]
	The result follows by noting that the quadratic variation $\langle  \log Z^\vpi \rangle_t$ is given by:
	\begin{align*}
		\langle  \log Z^\vpi \rangle_t = \sum_{i,j=1}^{n} \pi^i_t \pi^j_t  \left\langle \log X^i, \log X^j \right\rangle_t\		
		= \vpi_t^\intercal \mSigma_t \vpi_t
 	\end{align*}
 	and then rearranging terms.
\end{proof}

\subsection{Market Model Examples}

\cite{alaradi2018} assume growth rates and volatilities are bounded, differentiable, deterministic functions and the only driver of asset prices was a multidimensional Wiener process. In this section we present two market models satisfying the assumptions in this paper that allow for more general asset growth rates. The two models are presented with the goal of improved growth rate estimation in mind.

\subsubsection{Diffusion-Switching Growth Rate Process}

The \textbf{diffusion-switching model} assumes that asset growth rates switch between a number of possible diffusion processes according to an underlying Markov chain. That is:
\begin{equation}
\vgamma_t = \vgamma^{(\Theta_t)}_t\,,
\end{equation}
where $\Theta = (\Theta_t)_{t \geq 0}$ is a continuous-time Markov chain with state space $\mfM \coloneqq \{1,...,m\}$ and $\vgamma^{(i)}_t$ is the growth rate diffusion process associated with state $i\in\mfM$ given as the solution to the SDE:
\begin{equation}
d\mDim{\vphantom{\vW^\vgamma_t} \vgamma^{(i)}_t}{n}{1} = \mDim{\vphantom{\vW^\vgamma_t} \vphi}{n}{1}(t,\vgamma_t,i) ~dt + \mDim{\vphantom{\vW^\vgamma_t} \vPhi}{n}{k}(t,\vgamma_t,i) ~d\mDim{\vW^\vgamma_t}{k}{1}\,.
\end{equation}
In this formulation, $\vW^\vgamma$ is a $k$-dimensional Wiener process driving the growth rate diffusions and $\vphi$ and $\vPhi$ are the drift and volatility functions of the growth rate. We require $\vphi$ and $\vPhi$ to be chosen so that $\vgamma^{(i)} \in \LL^2$ for all $i$. A sufficient set of conditions for this are the usual Lipschitz and polynomial growth conditions that guarantee the existence of a unique, square-integrable strong solution to the SDE (see Theorem 2.9 in Chapter 5 of \cite{karatzas2012}). Figure \ref{fig:diffusionSwitchingPlot}  shows a simulation of this process when the possible diffusions are Ornstein-Uhlenbeck (OU) processes.

In Section \ref{sec:HMM}, we take both $\vphi$ and $\vPhi$ to be identically zero. This recovers the hidden Markov model (HMM) used in \cite{rieder2005},
where the growth rate switches between a number of possible constants rather than diffusion processes. This simplifies the calibration process and this is the model we employ in the implementation.

\begin{figure}[h!]
	\centering
	\includegraphics[width=0.45\textwidth]{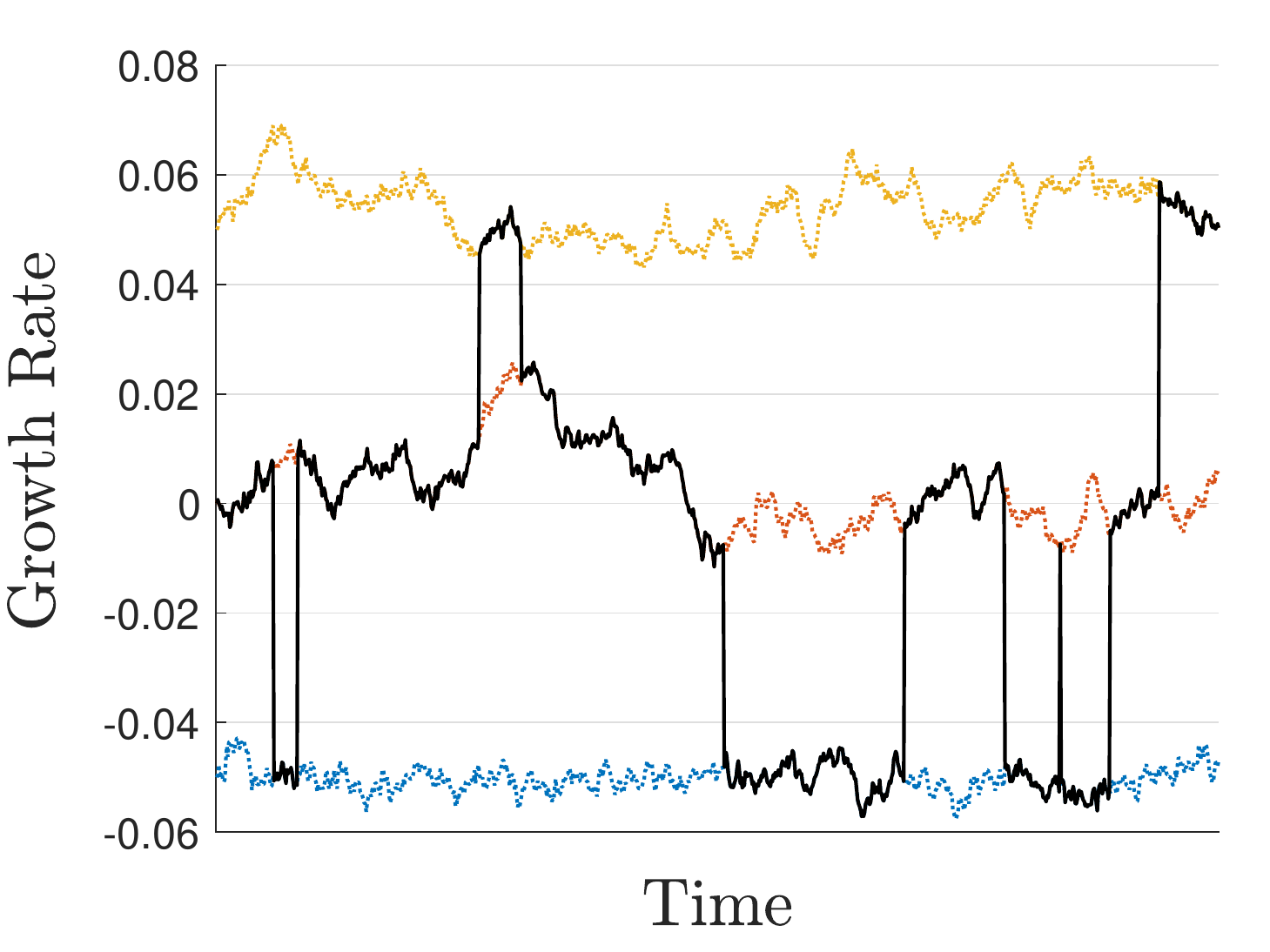}
	\captionsetup{width=.9\linewidth}
	\caption{Example of diffusion-switching process with $m=3$ states. The dotted colored lines represent the three possible diffusion processes the growth rate can follow corresponding to the three Markov chain states. The solid black line shows the growth rate path which jumps whenever a transition in the underlying Markov chain occurs.
}
	\label{fig:diffusionSwitchingPlot}
\end{figure}

\subsubsection{Second-Order Rank-Based Model}

An alternative model that may be considered is the \textbf{second-order rank-based model} of equity markets as described in \cite{fernholz2013}. In this model, an asset's price dynamics depend on the rank of the asset's market weights; typically, smaller assets have higher growth rates and volatilities than larger assets. The goal of this modeling approach is to better capture observed long-term characteristics of capital distribution in equity markets, such as average rank occupation times, by exploiting the inherent stability in the capital distribution curve.

Let $r^i_t$ be the rank of asset $i$ at time $t$, the asset price is assumed to satisfy the SDE:
\begin{equation}
d \log X^i_t = \left(\gamma^i + \sum_{j=1}^{n} g^j \ind_{\{r^i_t = j\}}\right) ~dt + \sum_{j=1}^{n} \sigma^j \ind_{\{r^i_t = j\}} ~dW^j_t
\end{equation}
That is, $\gamma^i$ is the ``name''-based growth rate of asset $i$ and $g^j$ is the additional growth an asset experiences when its capitalization occupies rank $j$. Similarly, $\sigma^j$ is the volatility of the asset in rank $j$. We assume the model parameters satisfy the requirements for the market to form an asymptotically stable system; see \cite{fernholz2013}, which also provides an outline for parameter estimation for this class of models.

It is important to notice that when this model is assumed, the rank processes for each of the stocks must be incorporated in the optimization problem as state variables. This can vastly complicate the proof of optimality when using a dynamic programming approach. The approach we take in the present work does not suffer from these issues involving local times and non-differentiability. Finally, we note that it is possible to create a hybrid model that is rank-dependent and driven by an unobservable Markov chain, but this may lead to difficulties in the parameter estimation.

\section{Stochastic Control Problem} \label{sec:control}

\subsection{Description}

The stochastic control problem we consider is similar to the one posed in \cite{alaradi2018}. The investor fixes two portfolios against which they measure their outperformance and their active risk, respectively. That is, the investor chooses a \textbf{performance benchmark} $\vrho$, which they wish to outperform, and a \textbf{tracking benchmark} $\veta$, which they penalize deviations from. The objective is to determine the portfolio process $\vpi$ that maximizes the expected growth rate differential relative to $\vrho$ over the investment horizon $T$. Moreover, the investor is penalized for taking on excessive levels of active risk (measured against $\veta$). An additional penalty independent of the two benchmarks is also included to control absolute risk (as measured by quadratic variation of wealth) or penalize allocation to certain assets as discussed in Section 4 of \cite{alaradi2018}.

The main state variable in our optimization problem is the logarithm of the ratio of the wealth of an arbitrary portfolio $\vpi$ relative to a preselected performance benchmark $\vrho$. Let $Y^{\vpi,\vrho}_t = \log \left(\frac {Z^\vpi_t}{Z^\vrho_t} \right)$ denote the logarithm of relative portfolio wealth for the portfolios $\vpi$ and $\vrho$. Then this process satisfies the SDE:
\begin{equation}\label{prop:Y_SDE}
dY^{\vpi,\vrho}_t = \left(\gamma^\vpi_t - \gamma^\vrho_t \right) dt + \left( \vpi _t - \vrho_t \right)^\intercal d\vM_t\;,
\end{equation}
which in turn implies
\begin{equation} \label{eqn:relWealth}
Y^{\vpi,\vrho}_t = Y^{\vpi,\vrho}_0 + \int_0^T \left(\gamma^\vpi_t - \gamma^\vrho_t \right) dt + \int_0^T \left( \vpi _t - \vrho_t \right)^\intercal d\vM_t\;.
\end{equation}

Our main \textbf{stochastic control problem} is to find the optimal portfolio $\vpi^*$ which, if the supremum is attained in the set of admissible strategies, achieves
\begin{equation} \label{eqn:control}
\underset{\vpi \in \As^c}{\sup} ~ H(\vpi)
\end{equation}
where $H(\vpi)$ is the \textbf{performance criteria} of a portfolio $\vpi$ given by:
{
\small
\begin{align} \label{eqn:perfCrit}
H(\vpi) = \EE \left[ \zeta^0 ~ Y^{\vpi,\vrho}_t - \tfrac{1}{2} \int_0^T \zeta^1_s ~ (\vpi_s - \veta_s)^\intercal \mOmega_s (\vpi_s - \veta_s) ~ ds  - \tfrac{1}{2}  \int_0^T \zeta^2_s ~ \vpi_s^\intercal \mQ_s \vpi_s ~ds \right]\;.
\end{align}
}
Here, $\zeta^0$ is a constant and $\vzeta = \left( \vzeta_t\right)_{t \geq 0}$ with $\vzeta_t = \left( \zeta^0, \zeta^1_t, \zeta^2_t \right)$ is an $\FF$-adapted process defined on $[0,\overline{\zeta}]^3$ for some fixed $\overline{\zeta} < \infty$. The vector $\vzeta_t$ represents the subjective preference parameters set by the investors to reflect their emphasis on three goals:

\begin{enumerate}
	\item The first term is a terminal reward term which corresponds to the investor wishing to \textbf{maximize the expected growth rate differential between their portfolio and the performance benchmark $\vrho$}. It is also equivalent to \textbf{maximizing the expected utility of relative wealth assuming a log-utility function}.
	\item The second term is a running penalty term which \textbf{penalizes deviations from the tracking benchmark}. When $\mOmega_t = \mSigma_t$, the investor is penalizing risk-weighted deviations from the tracking benchmark, with deviations in riskier assets being penalized more heavily. Thus, this can be seen as the investor aiming to \textbf{minimize tracking error/active risk}.
	\item The final term is a \textbf{general quadratic running penalty term} that does not involve either benchmark. One possible choice for $\mQ_t$ is the covariance matrix $\mSigma_t$, which can be adopted to minimize the \textit{absolute} risk of the portfolio measured in terms of the quadratic variation of the portfolio wealth process, $Z^\vpi_t$. Another option is to let $\mQ_t$ be a constant diagonal matrix, which has the effect of penalizing allocation in each asset according to the magnitude of the corresponding diagonal entry. The investor can use this choice of $\mQ$ as a way of imposing a set of ``soft'' constraints on allocation to each asset.
\end{enumerate}

The reader is referred to \cite{alaradi2018} for further interpretation of these terms.

\begin{Remark}
The two preference parameters $\zeta_{t}^{1,2}$ can be stochastic; e.g., they may depend on the investor's wealth level or other factors. Furthermore, the preference parameters are restricted to $[0,\overline{\zeta}]^3$ for two reasons: firstly, it simplifies the proof of optimality; secondly, from \cite{alaradi2018}, the results are driven by the relative weights, rather than absolute weights, therefore restriction to the cube results in no loss of generality.
\end{Remark}

\begin{Remark}
The benchmarks may be non-Markovian; if they are Markovian and can be represented as $\vrho_t=\rho(t,\vX_t)$ and $\veta_t=\eta(t,\vX_t)$, the functions $\rho$ and $\eta$ are not restricted to be differentiable. This allows for a much wider class of benchmarks including rank-based portfolios and portfolios constructed using additional information not related to asset prices, e.g., factor portfolios based on company fundamentals. Benchmarks from the class of functionally generated portfolios are allowed, including the market portfolio, as well as portfolios generated by rank-dependent portfolio generating functions, such as large-cap portfolios.
\end{Remark}

We also require the following assumption on the relative and absolute penalty matrices $\mOmega$ and $\mQ$:
\begin{Assumption}
The penalty matrices $\mOmega$ and $\mQ$ are $\FF$-adapted matrix-valued stochastic processes such that, for each $\mathbf{x} \in \RR^n$, there exists constants $\varepsilon > 0$ and $C < \infty$ satisfying
\begin{align}
\varepsilon \| \mx \|^2 ~\leq~ \mx^\intercal\, \mOmega_t\, \mx ~\leq~ C \| \mx \|,~
\qquad
\text{and}
\qquad
\varepsilon \| \mx \|^2 ~\leq~ \mx^\intercal\, \mQ_t \,\mx   ~\leq~ C \| \mx \|\,,
\qquad \forall t\ge0.
\end{align}
\end{Assumption}

These bounds play an analogous role to the nondegeneracy and bounded variance assumptions made on the quadratic covariation $\mSigma$, and ensure that the candidate optimal control we derive later is in fact admissible.

Allowing for stochastic penalty matrices is useful as it opens the door for stochastic volatility models in the case of $\mOmega$ (when choosing $\mOmega = \mSigma$) and stochastic transaction costs in the case of $\mQ$.

We next rewrite the control problem in terms of running reward/penalty terms. When either of the conditions in Assumption \ref{asmp:growthMtg} is enforced, the expected value of the last integral in \eqref{eqn:relWealth} is zero as the stochastic integral is in fact a martingale. Further, assuming that $Z^\vpi_0 = Z^\vrho_0$, the performance criteria becomes
\begin{align} \label{eqn:perfCrit2} \small
H(\vpi) = \EE \left[ \left\{\int_0^T \zeta^0 \left( \gamma^\vpi_t - \gamma^\vrho_t \right) - \tfrac{1}{2} \zeta^1_t ~ (\vpi_t - \veta_t)^\intercal \mOmega_t (\vpi_t - \veta_t) - \tfrac{1}{2} \zeta^2_t ~ \vpi_t^\intercal \mQ_t \vpi_t \right\} dt \right]\,.
\end{align}

The generalizations achieved thus far compared to \cite{alaradi2018} are summarized in Table \ref{tab:generalizations} below.
\begin{table}[h] \footnotesize \centering
\renewcommand{\arraystretch}{0.1}
	\begin{tabular}{ C{3cm} C{6cm} C{6cm}}
		\toprule[1.5pt]
		\newline & \textbf{Dynamic Programming} & \textbf{Convex Analysis}
		\\
		\midrule
		\newline
		\bf{Growth Rates}
		\newline
		&
		Bounded, deterministic, \newline differentiable growth rates  	
		& {Stochastic, unobservable growth rates (possibly rank-dependent)}
		\\
		\newline
		\bf{Noise component} \vspace{0.1cm}
		\newline
		&
		Deterministic, differentiable \newline volatility  with Brownian noise \vspace{0.1cm}
		& {$\LL^2$ (or even $\LL^1$) martingale noise \newline (possibly with stochastic volatility)}		\vspace{0.1cm}
		\\
		\newline
		\bf{Benchmarks} \vspace{0.3cm}
		\newline
		&
		Benchmarks are Markovian in $\vX$, \newline differentiable maps \vspace{0.3cm}
		& {Benchmarks are $\FF$-adapted}	\vspace{0.3cm}
		\\
		\newline
		\bf{Penalty matrices}
		\newline
		&
		Deterministic penalty \newline weighting matrices		
		& {Stochastic penalty \newline weighting matrices}		
		\\
		\newline
		\bf{Preference parameters}
		\newline
		&
		Constant subjective \newline preference parameters
		& {Stochastic subjective preference parameters (e.g. wealth-dependent)}
		\\
		\bottomrule[1.5pt]					
	\end{tabular}
	\captionsetup{width=.8\linewidth}
	\caption{Summary of generalizations achieved using the convex analysis approach over the dynamic programming approach.}
	\label{tab:generalizations}
\end{table}
\subsection{Projection}

To solve the the control problem \eqref{eqn:control} we follow the arguments in  \cite{casgrain2018}. The first step is to project the asset price dynamics onto the observable filtration $\mathfrak{F}$, which enables us to rewrite the performance criteria \eqref{eqn:perfCrit2} in terms of observable processes only. To this end, we first define the conditional expectation process $\vgammaHat = \left( \vgammaHat_t \right)_{t \geq 0}$ so that
\[  \vgammaHat_t \coloneqq \EE \left[ \vgamma_t ~\middle|~ \Ft_t \right]\,. \]
This process represents the investor's best estimate of the assets' growth rates given all asset price information up to a given point in time. Similarly, we define a process $\mSigmaHat =  ( \mSigmaHat_t )_{t \geq 0}$ corresponding to the projection of the unobservable quadratic covariation onto the same filtration, so that
\[ \mSigmaHat_t \coloneqq \EE \left[ \mSigma_t ~\middle|~ \Ft_t \right]\,. \]

\begin{Proposition} \label{prop:projection}
Regardless of which part of Assumption \ref{asmp:growthMtg} is enforced, the estimated growth rate process $\vgammaHat$ is $\FF$-adapted with $\vgammaHat \in \LL^2$. Furthermore, the estimated quadratic covariation process $\mSigmaHat$ is $\FF$-adapted and satisfies:
	\[
	\varepsilon \| \mathbf{x} \|^2 ~\leq~ \mathbf{x}^\intercal \mSigmaHat_t \mathbf{x}  ~\leq~ C \| \mathbf{x} \|^2
	\]
	for some $\varepsilon > 0$ and $C < \infty$ for all $\mathbf{x} \in \RR^n$ and $t \geq 0$.
\end{Proposition}
\begin{proof}
Assumption \ref{asmp:posDef} implies that all entries of $\mSigma$ are bounded (see Appendix A of \cite{alaradi2018}), and it follows that the conditional expectation $\mSigmaHat$ is well-defined, bounded and $\FF$-adapted. This also implies the required inequalities involving $\mSigmaHat$.

To prove the statements regarding $\vgammaHat$, first, under either condition of Assumption \ref{asmp:growthMtg}, we have that $\EE[|\vgamma_t|] < \infty$ for all $t$ and hence $\EE \left[ \vgamma_t ~\middle|~ \Ft_t \right]$ exists and is unique and integrable (see \cite{durrett2010probability}, Lemma 5.1.1). Furthermore, by the definition of conditional expectations, $\EE \left[ \vgamma_t ~\middle|~ \Ft_t \right]$ is $\Ft_t$-measurable for all $t \geq 0$.
	
Next, suppose Assumption \ref{asmp:growthMtg}(a) is enforced, so that $\vgamma \in \LL^2$, i.e. that $\gamma^i \in \LL^2$ for each $i$. By the same reasoning as above, this implies that $\EE \left[ (\gamma^i_t)^2 ~\middle|~ \Ft_t \right]$ exists, is unique, and integrable. By Jensen's inequality for conditional expectations (Theorem 5.1.3 of \cite{durrett2010probability})
\begin{align*}
\left(\EE \left[ \gamma^i_t ~\middle|~ \Ft_t \right]\right) ^2 ~&\leq~ \EE \left[ (\gamma^i_t)^2 ~\middle|~ \Ft_t \right]
\\
\implies ~~~  \EE \left[ \left(\EE \left[ \gamma^i_t ~\middle|~ \Ft_t \right]\right) ^2 \right] ~&\leq~  \EE \left[ \EE \left[ (\gamma^i_t)^2 ~\middle|~ \Ft_t \right] \right]
\\
\implies ~~~  \EE \left[ (\gammaHat^i_t) ^2 \right] ~&\leq~  \EE \left[ (\gamma^i_t)^2 \right]  < \infty
\end{align*}
As this is true for each $i$, it follows that $\vgammaHat \in \LL^2$.

Finally, suppose Assumption \ref{asmp:growthMtg}(b) is enforced, so that $\vgamma\in{\LL^\infty,M}$, then we have that $\vgamma \in \LL^2$, and the conclusion follows.
\end{proof}

The projection is measurable with respect to the investor's filtration, and hence may be used to construct their portfolio.
\begin{Lemma} \label{lemma:innovations}
	The innovations process, $\vMHat = \left( \vMHat_t \right)_{t \geq 0}$, defined by
	\begin{equation}
		\vMHat_t \coloneqq \log \vX_t - \int_0^t \vgammaHat_s ~ds
	\end{equation}
	is an $\FF$-adapted martingale with $\vMHat \in \LL^1$. Furthermore, the asset dynamics satisfies an SDE in terms of $\FF$-adapted processes as follows
	\begin{equation}
		d \log \vX_t = \vgammaHat_t ~ dt + d \vMHat_t\,.
	\end{equation}
\end{Lemma}
\begin{proof} See Appendix \ref{proof:innovations}.
\end{proof}
 	
\subsection{Optimization via Convex Analysis}

The performance criterion \eqref{eqn:perfCrit2} may be written in terms of the projected processes, so that
\begin{equation*} \small
H(\vpi) = \EE \left[ \int_0^T \Big\{ \zeta^0 \left( \gammaHat^\vpi_t - \gammaHat^\vrho_t \right) - \tfrac{1}{2} \zeta^1_t ~ (\vpi_t - \veta_t)^\intercal \mOmega_t (\vpi_t - \veta_t) - \tfrac{1}{2} \zeta^2_t ~ \vpi_t^\intercal \mQ_t \vpi_t \Big\} ~dt \right]
\end{equation*}
where $\vgamma$ and $\mSigma$ are replaced with their conditional expectations $\vgammaHat$ and $\mSigmaHat$ and where $\gammaHat^\vpi$ is the projected growth rate of portfolio $\vpi$ defined analogously to the portfolio growth rate given in Proposition \ref{prop:portfolioDynamics}. This replacement is justified by Lemma \ref{lemma:innovations} and the fact that we have $\EE\left[ \mSigma_t \right] = \EE \left[ \EE \left[ \mSigma_t \middle| \Ft_t \right] \right] = \EE[ \mSigmaHat_t ] $ and similarly $\EE\left[ \text{diag}(\mSigma_t) \right] = \EE[ \text{diag}(\mSigmaHat_t) ] $. Following the steps in the proof of Proposition 2 of \cite{alaradi2018}, the performance criteria can be written in the following linear-quadratic form: {\small
\begin{align*}
H(\vpi) = \EE \biggl[ &  \int_0^T  \left\{- \tfrac{1}{2} \vpi_t^\intercal \mA_t \vpi_t + \vpi_t^\intercal \vB_t \right\} dt \biggr] - \EE \left[ \int_0^T \left\{ \zeta^0 \widehat{\gamma}^\vrho_t -  \tfrac{1}{2} \zeta^1_t\,\veta_t^\intercal \mOmega_t \veta_t \right\} dt \right]
\\
\text{where } \qquad \mA_t &=  \zeta^0 \mSigmaHat_t + \zeta^1_t \mOmega_t + \zeta^2_t \mQ_t \nonumber
\\
\vB_t &= \zeta^0 \left( \widehat{\vgamma}_t + \tfrac{1}{2} \mbox{diag}(\mSigmaHat_t) \right) + \zeta^1_t \mOmega_t \veta_t \nonumber
\end{align*}
}%

As the second expectation does not depend on the control, it may be omitted in the optimization. It is convenient to define the assets' instantaneous rate of return process $\valpha = \left(\valpha_t\right)_{t \geq 0}$ given by
\begin{equation}
\valpha_t \coloneqq \vgamma_t + \tfrac{1}{2} \mbox{diag}(\mSigma_t)\,,
\end{equation}
along with its projected counterpart $\valphaHat_t = \vgammaHat_t + \tfrac{1}{2} \mbox{diag}(\mSigmaHat_t)$. With this, we can write the performance criteria that we aim to optimize as
\begin{subequations}
 \label{eqn:perfCrit5}
\begin{equation}
H(\vpi) = \EE \left[ \int_0^T  \left\{- \tfrac{1}{2} \vpi_t^\intercal \mA_t \vpi_t + \vpi_t^\intercal \vB_t \right\} dt \right]\,,
\end{equation}
where
\begin{align}
\mA_t &=  \zeta^0 \mSigmaHat_t + \zeta^1_t \mOmega_t + \zeta^2_t \mQ_t
\\
\vB_t &= \zeta^0 \valphaHat_t + \zeta^1_t \mOmega_t \veta_t
\end{align}
\end{subequations}%

Next, we set our (unconstrained) search space to be
\begin{align*}
\As = \left\{ \vpi: \Omega \times [0,T] \rightarrow \RR^n, \vpi\in\LL^2,\;\FF\text{-adapted} \right\}\,.
\end{align*}
This forms a reflexive Banach space with norm $\| \vpi \|_{\LL^2} = \left( \EE \left[\int_0^T \| \vpi_t \|^2 ~dt \right]\right)^{1/2}$ - see Theorem 1.3 and Theorem 4.1 of \cite{stein2011functional}. The admissible set we wish to optimize over (the set of admissible portfolios $\As^c$) is a subset of $\As$ whether portfolios are defined to be bounded or just $\LL^2$. Let $\As^*$ be the dual space and let $\langle \cdot, \cdot \rangle$ denote the canonical bilinear pairing over $\As \times \As^*$, i.e. $\langle u, u^* \rangle = u^*(u)$. The performance criteria \eqref{eqn:perfCrit5} can be viewed as a functional that maps elements from $\As$ to a real number, i.e. $H: \As \rightarrow \RR$. The following proposition ensures that the candidate optimal control we propose later is indeed optimal.

\begin{Proposition} \label{prop:concave}
The functional $H: \As \rightarrow \RR$ given by \eqref{eqn:perfCrit5} is proper, upper semi-continuous and strictly concave.
\end{Proposition}
\begin{proof}
See Appendix \ref{proof:concave}.
\end{proof} 

The subsequent steps are to (i) compute the G\^{a}teaux differential associated with the functional $H$, (ii) find an element in the admissible set which makes it vanish, and (iii) use the strict concavity of the objective function to conclude that this element is a global maximizer. For more further details, the reader is referred to the first two chapters of \cite{ekeland1999convex}. 
\begin{Proposition} \label{prop:gateaux}
The functional $H: \As \rightarrow \RR$ is G\^{a}teaux differentiable for all  $\widetilde{\vpi}, \vpi \in \As$ with G\^{a}teaux differential $H'(\vpi) \in \As^*$ given by
\begin{equation}
\left\langle \widetilde{\vpi} , H'\left(\vpi\right) \right\rangle = \EE \left[ \int_0^T  \widetilde{\vpi}_t^\intercal \left[ -\mA_t \vpi_t + \vB_t \right] ~dt \right]\,
\end{equation}
\end{Proposition}
\begin{proof}
See Appendix \ref{proof:gateaux}.
\end{proof}

Finally, we present our main result which gives the form of the optimal control for the stochastic control problem \eqref{eqn:control}.
\begin{Theorem} \label{thm:optCont}
The portfolio given by
\begin{subequations}
\label{eqn:opt-pi-A-and-B}
\begin{equation} \label{eqn:optCont}
\vpi^*_t = \mA^{-1}_t ~ \left[ \frac{1 - \ones^\intercal \mA^{-1}_t  \vB_t}{\ones^\intercal \mA^{-1}_t \ones} ~ \ones + \vB_t \right]\;,
\end{equation} where
{\small
\begin{equation}
\mA_t =  \zeta^0 \mSigmaHat_t + \zeta^1_t \mOmega_t + \zeta^2_t \mQ_t\;, \quad
\text{and} \quad
\vB_t = \zeta^0 \valphaHat_t + \zeta^1_t \mOmega_t \veta_t \;.
\end{equation}
}
\end{subequations} is the unique solution to the stochastic control problem \eqref{eqn:control}.
\end{Theorem}
\begin{proof}
See Appendix \ref{proof:optCont}.
\end{proof}

The optimal portfolio given by the above theorem resembles the solution in \cite{alaradi2018}, but the unobservable growth rate $\vgamma$ and quadratic covariation matrix $\mSigma$ are replaced by their projected counterparts $\vgammaHat$ and $\mSigmaHat$.

Furthermore, we can relate this optimal solution to the growth optimal portfolio (GOP) and the minimum quadratic variation portfolio (MQP). To recall: the GOP is the portfolio with maximal expected growth over any time horizon while the MQP is the portfolio with the smallest quadratic variation of all portfolios over the investment horizon. Formally, the GOP and MQP are the solutions to the  optimization problems
\begin{align} \label{GOPproblem}
\vpi^{GOP} &= \underset{\vpi \in \As_c}{\arg \sup} ~ \EE \left[ \log \left( \frac{Z^\vpi_T}{Z^\vpi_0} \right) \right]  \qquad\text{and} \\
\vpi^{MQP} &= \underset{\vpi \in \As_c}{\arg \inf} ~~ \EE \left[ \int_0^T \vpi_t^\intercal \mSigma_t \vpi_t ~dt \right] \;,\end{align}
respectively.
\begin{Corollary} \label{cor:GOPMQP}
The growth optimal portfolio is given by:
\begin{equation} \label{eqn:GOP}
\vpi^{GOP}_t = \left( 1 - \ones^\intercal \mSigmaHat^{-1}_t \valphaHat_t \right)  \vpi^{MQP}_t + \mSigmaHat^{-1}_t \valphaHat_t \;.
\end{equation}
where $\vpi^{MQP}$ is the minimum quadratic variation portfolio given by:
\begin{equation} \label{eqn:MQP}
\vpi^{MQP}_t = \frac{1}{\ones^\intercal \mSigmaHat^{-1}_t \ones} ~ \mSigmaHat^{-1}_t \ones \;.
\end{equation}
\end{Corollary}
\begin{proof}
Use the optimal control given in Theorem \ref{thm:optCont} with $\zeta^0 = \zeta^1 = 0$ and $\mQ = \mSigma$ to obtain the MQP, and with $\zeta^1 = \zeta^2 = 0$ to obtain the GOP.
\end{proof} \\

Next, we show how the optimal portfolio may be split into sub-portfolios and how it may be written as a growth optimal portfolio for a modified asset price model.
\begin{Corollary} \label{cor:properties}
When $\mOmega_t = \mQ_t = \mSigma_t$ (minimize relative and absolute risk), the optimal portfolio in Theorem \ref{thm:optCont} can be represented as
\[ \vpi^*_t = c^0_t ~ \vpi^{GOP}_t + c^1_t ~ \veta_t + c^2_t ~ \vpi^{MQP}_t\,, \]
where $c^i_t = \frac{\zeta^i_t}{\zeta^0 + \zeta^1_t + \zeta^2_t} \geq 0$ for $i =1,2,3$ so that $c_t^1 + c_t^2 + c_t^3 = 1$.
\end{Corollary}
\begin{proof}
The results follow by adapting the proof in Appendix A.3 of \cite{alaradi2018} to the current setting.
\end{proof}

\begin{Corollary} \label{cor:modifiedSigma}
The optimal portfolio $\vpi^*$ is the growth optimal portfolio for a market with a modified (projected) rate of return process $\valpha^*$ and a modified quadratic covariation matrix $\mSigma^*$ given by:
\begin{subequations}
\begin{align}
\valphaHat^*_t &= \zeta^0_t \valphaHat_t + \zeta^1_t \mSigma_t \veta_t
\\
\mSigma^*_t &= \zeta^0_t \mSigma_t + \zeta^1_t \mOmega_t + \zeta^2_t \mQ_t
\end{align}
\end{subequations}
\end{Corollary}
\begin{proof}
The result follows from direct comparison of the optimal portfolio in Theorem \ref{thm:optCont} with the form of the growth optimal portfolio in Corollary \ref{cor:GOPMQP}.
\end{proof}


A few comments on the last two results:
\begin{enumerate}	
	
	\item When the investor wishes to minimize both relative and absolute risk, their optimal strategy is to invest in the GOP, MQP and tracking benchmark. The idea is to benefit from the expected growth rate of the GOP while modulating its high levels of absolute and active risk by investing in the MQP and tracking benchmark, respectively. The proportions are determined by the relative importance the investor places on outperformance, tracking and absolute risk as represented by $\vzeta$ and may vary stochastically through time. Notice also that the optimal solution does not depend on the performance benchmark $\vrho$ and is myopic, i.e. independent of the investment horizon $T$. This is expected as the GOP can be shown to maximize expected growth at \textit{any} time horizon.
	
	\item The absolute penalty term forces the optimal strategy towards the ``shrinkage'' portfolio given by $\frac{1}{\ones' \mQ^{-1}_t \ones} ~ \mQ^{-1}_t \ones$. When $\mQ$ is a diagonal matrix $\mQ_t = \text{diag}\left(w^1_t,...,w^n_t \right)$, then $\mQ^{-1}_t = \text{diag}\left(\tfrac{1}{w^1_t},...,\tfrac{1}{w^n_t} \right)$. From this, we see that the absolute penalty term forces us to shrink to a portfolio proportional to $\left(\tfrac{1}{w^1_t},...,\tfrac{1}{w^n_t} \right)$; when $w_i$ is large $1/w_i$ is small and the shrinkage portfolio allocates less capital to asset $i$. This can be used to tilt away from undesired assets; taking $w_i \rightarrow \infty$ forces the allocation in asset $i$ to zero. Additionally, taking $\mQ = \mI$ penalizes large positions in any asset by shrinking to the equal-weight portfolio, while setting $w^i = \mSigma_{ii}$ forces shrinkage to the risk parity portfolio. 	
	
	\item The second corollary implies that the investor is modifying the assets' rates of return to reward those assets that are more closely correlated with the portfolio they are trying to track. This is because the term $\zeta_1 \mSigma_t \eta_t$ is a vector consisting of the quadratic covariation between between each asset and the wealth process associated with the tracking benchmark $\veta$, i.e. $\langle \log X_i, \log Z_\veta \rangle$. Moreover, if we consider the case where $\mQ_t$ is a diagonal matrix, the modification to the covariance matrix amounts to increasing the variance of each asset according to the corresponding diagonal entry of $\mQ_t$. This in turn makes certain assets less desirable and is tied to the notion of tilting away from those assets as discussed in the previous point.
	
\end{enumerate}

The points made above highlight the motivation for including the two running penalties. In principle, when an investor's goal is simply to outperform a performance benchmark, they would hold the GOP to maximize their expected growth. However, it is well-documented that the GOP is associated with very large levels of risk (in terms of portfolio variance) as well as potentially large short positions in a number of assets. Adding the relative and absolute penalty terms mitigates some of this risk. It is also worth noting that the decompositions in Corollary 3 (ii) and (iii) can help guide the choice of the subjective parameters $\zeta_i$ which can be used to express the proportion of wealth the investor wishes to place in the GOP, tracking benchmark and MQP.

\section{Hidden Markov Model} \label{sec:HMM}

The model presented above is quite general but not implementable without an explicit form for the conditional expectations $\vgammaHat$ and $\mSigmaHat$. In this section, we compute these conditional expectations for a hidden Markov model where the growth rate switches between a number of possible constant vectors according to an underlying Markov chain.

More specifically, let $\vW = \left\{ \vW_t = (W^1_t,..., W^k_t)^\intercal \right\}_{t \geq 0}$ be a standard $k$-dimensional Wiener process and let $\Theta = (\Theta_t)_{t\geq 0}$  an unobservable, continuous-time Markov chain with state space $\mfM = \{1,2,...,m\}$ and generator matrix $\mG$. Next, suppose the asset prices satisfies the SDE
\begin{equation}  \label{eqn:hmm}
d \log \vX_t = \vgamma^{(\Theta_t)} ~dt + \vxi ~d\vW_t\,,
\end{equation}
where $\vgamma^{(j)} \in \RR^n$ for $j = 1,..., m$ and $\vxi$ is a constant matrix of sensitivities. Here, the growth rates are all bounded and the noise component is a square-integrable martingale.

\begin{Remark}
To stabilize the estimation process, we choose $\vxi$ to be time-independent. Moreover, $\vxi$ cannot depend on the Markov chain in the same manner as the growth rates (i.e. switch between a number of fixed matrices according to the prevailing Markov chain state), otherwise the Markov chain becomes observable and there is no filtering problem; see \cite{krishnamurthy2016filterbased}.
\end{Remark}

If $\mSigma = \vxi\,\vxi^\intercal$ is a constant matrix that satisfies Assumption \ref{asmp:posDef}, then the market model above satisfies the requirements described earlier in the paper and the optimal portfolio is given by \eqref{eqn:opt-pi-A-and-B}. Here, $\mSigmaHat_t = \mSigma$ for all $t$. It remains to compute an explicit form for $\vgammaHat_t$. This filtering step involves deriving a posterior distribution for the current state of the Markov chain given all the observable information up to the present time, i.e., we are interested in computing $\PP(\Theta_t|\Ft_t)$.

First, let the investor's prior distribution for the initial state of the Markov chain be denoted $p^j_0 \coloneqq \PP(\Theta_0 = j)$ $\forall\,j \in\mfM$. Next, we define the posterior distribution of the state of the underlying Markov chain given the observable data as
\begin{equation}
p^j_t = \EE\left[ \ind_{\{\Theta_t = j\}} ~\middle|~ \Ft_t  \right] ~ , \qquad \forall \;j \in\mfM\,.
\end{equation}
The following lemma allows us to write the posterior probabilities in terms of un-normalized state variables.
\begin{Lemma} \label{lemma:unnormalizedPosterior}
Let $\vlambda = \left( \vlambda_t \right)_{t\in[0,T]}$ be an $\FF$-adapted process satisfying Novikov's condition
\[ \EE \left[ \exp \left( \tfrac{1}{2} \int_0^T \| \vlambda_t \|^2 ~dt \right) \right] < \infty \, , \] 
and define: 
\begin{enumerate}
	\item A probability measure $\widetilde{\PP}$ via the Radon-Nikodym derivative
	\begin{equation*}
	\frac{d \widetilde{\PP}}{d\PP} = \exp\left[ - \int_0^T \vlambda_{t^-}^\intercal ~d\vW_t - \frac{1}{2} \int_0^T \vlambda_{t^-}^\intercal \vlambda_{t^-} dt \right]\,.
	\end{equation*}
	\item A stochastic process $\Upsilon = \left( \Upsilon_t \right)_{t \geq 0}$
	\[ \Upsilon_t = \EE\left[ \frac{d\PP}{d \widetilde{\PP}} ~\middle|~ \Ft_t  \right] = \exp\left[ \int_0^t \vlambda_{t^-}^\intercal ~d\vW_t - \frac{1}{2} \int_0^t \vlambda_{t^-}^\intercal \vlambda_{t^-} dt \right] \,. \]
\end{enumerate}  
Then, the posterior probability processes $p_t^j$ can be written as:
\begin{equation}
p^j_t = \frac{P^j_t}{\sum_{i\in\mfM} P^i_t} , \qquad \forall \;j\in\mfM\,,
\end{equation}
where $P_t^j = \EE^{\widetilde{\PP}} \left[ \ind_{\{\Theta_t = j\}} \Upsilon_t ~\middle|~ \Ft_t \right]$.
\end{Lemma}
\begin{proof}
	See Appendix \ref{proof:unnormalizedPosterior}.
\end{proof} \\

The main theorem for this section gives the SDE system that governs the dynamics of the posterior probabilities.
\begin{Theorem} \label{thm:filter}	
The state variables $\left\{ P_j \right\}_{j\in\mfM}$ satisfy the following stochastic differential equations:
\begin{equation} \label{eqn:filterSDE}
dP^j_t ~=~  \sum_{i\in\mfM} P^i_{t^-} \mG_{ji} ~dt + P^j_{t^-}\,\left(\vgamma^{(j)}\right)^\intercal \mSigma^{-1} ~ d \log \vX_t\,,
\end{equation}
with initial conditions $P^j_0 = p^j_0$ and $\vgamma^{(j)} = \left(\gamma^{(j)}_1, ..., \gamma^{(j)}_n \right)^\intercal$.
\end{Theorem}
\begin{proof}
See Appendix \ref{proof:filter1}.
\end{proof} \\

In the sequel, we use the vector notation $\vp_t = \left(p^1_t,\dots, p^m_t \right)^\intercal$ and $\vP_t = \left(P^1_t,\dots, P^m_t \right)^\intercal$. Note also that $\vp = ( \vp_t )_{t \geq 0}$ and $\vP = ( \vP_t )_{t \geq 0}$ are both $\FF$-adapted stochastic processes.

Using the result above we arrive at the final form for the projected asset growth rates as
\begin{equation}
\vgammaHat_t = \sum_{j=1}^m p^j_t \vgamma^{(j)}\,.
\end{equation}

In this model, the optimal portfolio \eqref{eqn:opt-pi-A-and-B} is a weighted average of optimal portfolios corresponding to each state of the Markov chain. In particular, the state-dependent portfolios are the optimal solutions assuming the Markov chain remains in the corresponding states at all times and the weights are given by the investor's posterior probabilities that a given state prevails. In other words, the portfolio \eqref{eqn:opt-pi-A-and-B} is the posterior mean of optimal portfolios across states. This is summarized in the following corollary.
\begin{Corollary}
For the HMM model \eqref{eqn:hmm} of asset price dynamics, the optimal portfolio \eqref{eqn:opt-pi-A-and-B} can be written as
\begin{equation*}
\vpi^*_t = \sum_{i\in\mfM} p^i_t \,\vpi^{(i)}_t\,,
\end{equation*}
\begin{align*}
\text{where } \quad \vpi^{(i)}_t &= \mA^{-1}_t \left[ \frac{1 - \ones^\intercal \mA^{-1}_t  \vB^{(i)}_t}{\ones^\intercal \mA^{-1}_t \ones} ~ \ones + \vB^{(i)}_t \right]\;,
\\
\text{with} \qquad \mA_t &=  \zeta^0 \mSigma + \zeta^1_t \mOmega_t + \zeta^2_t \mQ_t\;,
\\
\text{and} \qquad \vB^{(i)}_t &= \zeta^0 \vgamma^{(i)} + \zeta^1_t \mSigma \veta_t \;.
\end{align*} is the solution to the control problem \eqref{eqn:control} assuming the Markov chain state space is $\mathcal{S} = \{i\}$.
\end{Corollary}
\begin{proof}
The result follows as a consequence of the linearity of the optimal control with respect to the inferred growth rate $\vgammaHat$. Explicitly, we have
\begin{align*}
\sum_{i=1}^m p^i_t \vpi^{(i)}_t ~&=~ \sum_{i=1}^m p^i_t \mA^{-1}_t  \left[ \frac{1 + \ones^\intercal \mA^{-1}_t  \vB^{(i)}_t}{\ones^\intercal \mA^{-1}_t \ones} ~ \ones - \vB^{(i)}_t \right]
\\
~&=~ \mA^{-1}_t \sum_{i=1}^m p^i_t \left[ \frac{1 + \ones^\intercal \mA^{-1}_t  \vB^{(i)}_t}{\ones^\intercal \mA^{-1}_t \ones} ~ \ones - \vB^{(i)}_t \right]
\\
~&=~ \mA^{-1}_t  \left[ \frac{1 + \ones^\intercal \mA^{-1}_t \sum_{i=1}^m p^i_t  \vB^{(i)}_t}{\ones^\intercal \mA^{-1}_t \ones} ~ \ones - \sum_{i=1}^m p^i_t \vB^{(i)}_t \right]
\end{align*}
Now,
\begin{align*}
\sum_{i=1}^m p^i_t \vB^{(i)}_t ~&=~  \sum_{i=1}^m p^i_t \left[\zeta^0_t \vgamma^{(i)}_t + \zeta^1_t \mSigma \veta_t \right]
\\
~&=~  \zeta^0_t \vgammaHat_t + \zeta^1_t \mSigma \veta_t
\\
~&=~  \vB_t
\end{align*}
Substituting this sum into the previous expression completes the proof.
\end{proof}


\section{Implementation}

In this section, we test the performance of the optimal portfolio in a series of out-of-sample backtests. Before conducting these tests, we first derive a discretized version of the filter given in Theorem \ref{thm:filter} and outline the procedure for estimating the HMM parameters from data. Henceforth, we work with daily data and fix the time increment to be $\Delta t = \tfrac{1}{252}$.

\subsection{Data}

The data consist of daily returns (with dividends) for 12 industries (Non-durables, Durables, Manufacturing, Energy, Chemicals, Business Equipment, Telecom, Utilities, Shops, Health, Money and Other) for the period January 1970 to December 2017. We construct industry portfolios by assigning each stock to one of the 12 groupings listed above according to their SIC code at the time. The resulting industries constitute our market, i.e. our market consists of 12 assets that are the industries themselves. We use industry returns rather than individual securities to avoid difficulties arising from varying investment sets caused by individual companies entering and exiting the market. Furthermore, restricting to 12 assets makes the parameter estimation step more manageable; parameter estimation, particularly of the growth rates, becomes less reliable and robust as the number of assets increases.

\subsection{Model Discretization}

To proceed with implementing the filter in Theorem \ref{thm:filter}, we first discretize the SDE system \eqref{eqn:filterSDE}. A na\"{i}ve application of the Euler-Maruyama method does not yield acceptable results, due to instability issues that lead to underflow in implementation. Hence, we derive an alternative discretization scheme for the un-normalized posterior probabilities:
\begin{align} \label{eqn:discreteFilter}
\vP(t + \Delta t) &= \mDim{e^{\mG \Delta t}}{m}{m} ~ \mDim{\left(\begin{array}{c}
P^1_t \, Y_1(t,t + \Delta t) \\
\vdots \\
P^m_t \, Y_m(t,t + \Delta t)
\end{array}\right)}{m}{1}
\\
\mbox{where } ~~ Y_j(t,t + \Delta t) &= \exp \left\{ -\tfrac{1}{2} \vgamma^{(j)\intercal}\, \mSigma^{-1} \,\vgamma^{(j)} \Delta t + \vgamma^{(j)\intercal} \,\mSigma^{-1} \,\log \left( \tfrac{\vX_{t+\Delta t}}{\vX_t} \right) \right\}\,.\nonumber
\end{align}
In the last line above, the ratio $\tfrac{\vX_{t+\Delta t}}{\vX_t}$ is element-wise. We defer the full derivation of this scheme to Appendix \ref{proof:discreteFilter}.

Another aspect of discretizing the model involves finding the nearest generator matrix for a given transition probability matrix. As seen in the next section, our estimation procedure produces a transition probability matrix $\mZ$ which governs transitions in discrete time rather than a continuous time generator matrix $\mG$. This results from the time discretization. The optimal portfolio, however, requires the generator matrix which is trivially related to the discrete time transition probability matrix via $\mZ = e^{\mG \Delta t}$. And one is tempted to take the matrix logarithm to obtain $\mG$ from $\mZ$. However, this sometimes produces a matrix with complex entries which is unusable for our purposes and necessitates an alternative procedure, and is a well known problem in credit migration problems. We opt to use an estimated generator matrix $\mG^*$ as the one which minimizes the Frobenius norm to the probability transition matrix
\begin{equation}
\label{eqn:GeneratorFromTransition}
\mG^* = \underset{\mG}{\arg\min} \left\| \mZ - e^{\mG \Delta t} \right\|.
\end{equation}

\subsubsection{Parameter Estimation and Model Selection}

The parameter set for the HMM market model \eqref{eqn:hmm} consists of the constant growth rates in each state $\vgamma^{(j)}$ for $j =1,...,m$, the shared covariance matrix $\mSigma$ and the generator matrix $\mG$. We estimate these parameters using maximum likelihood estimation, and as our model contains latent information it requires using the Expectation-Maximization (EM) algorithm of \cite{baum1970maximization} (see also \cite{bishop2006pattern} for more details).
Our model contains Gaussian emissions from the latent states, hence we obtain explicit update rules for the Baum-Welch (forward-backward) algorithm (see Appendix \ref{sec:EM} for details). The EM algorithm is repeated multiple times with different initial values and the MLE is taken to be the parameter set with the largest likelihood from all the EM runs. The initialization for each run is based on estimating an independent mixture model (without a latent Markov chain component).

As noted above, the estimation algorithm produces the transition probability matrix for discrete time observations. This matrix is then transformed into a generator matrix using \eqref{eqn:GeneratorFromTransition}. Alternative estimation approaches that do not require assuming discrete time steps, such as the MCMC-based approach outlined in \cite{hahn2010markov}, may be used instead.

Figures \ref{fig:HMMparameterEstimation} and \ref{fig:HMMlatentFactorEstimation} and Table \ref{tab:viterbiFreq} show the estimation results for the five-year period 2010 to 2015 for the 12-industry granularity level. The results show the general pattern of estimated parameters using this approach summarized as follows:

\begin{landscape}
	\begin{figure}
		\centering
		\includegraphics[width=1\linewidth]{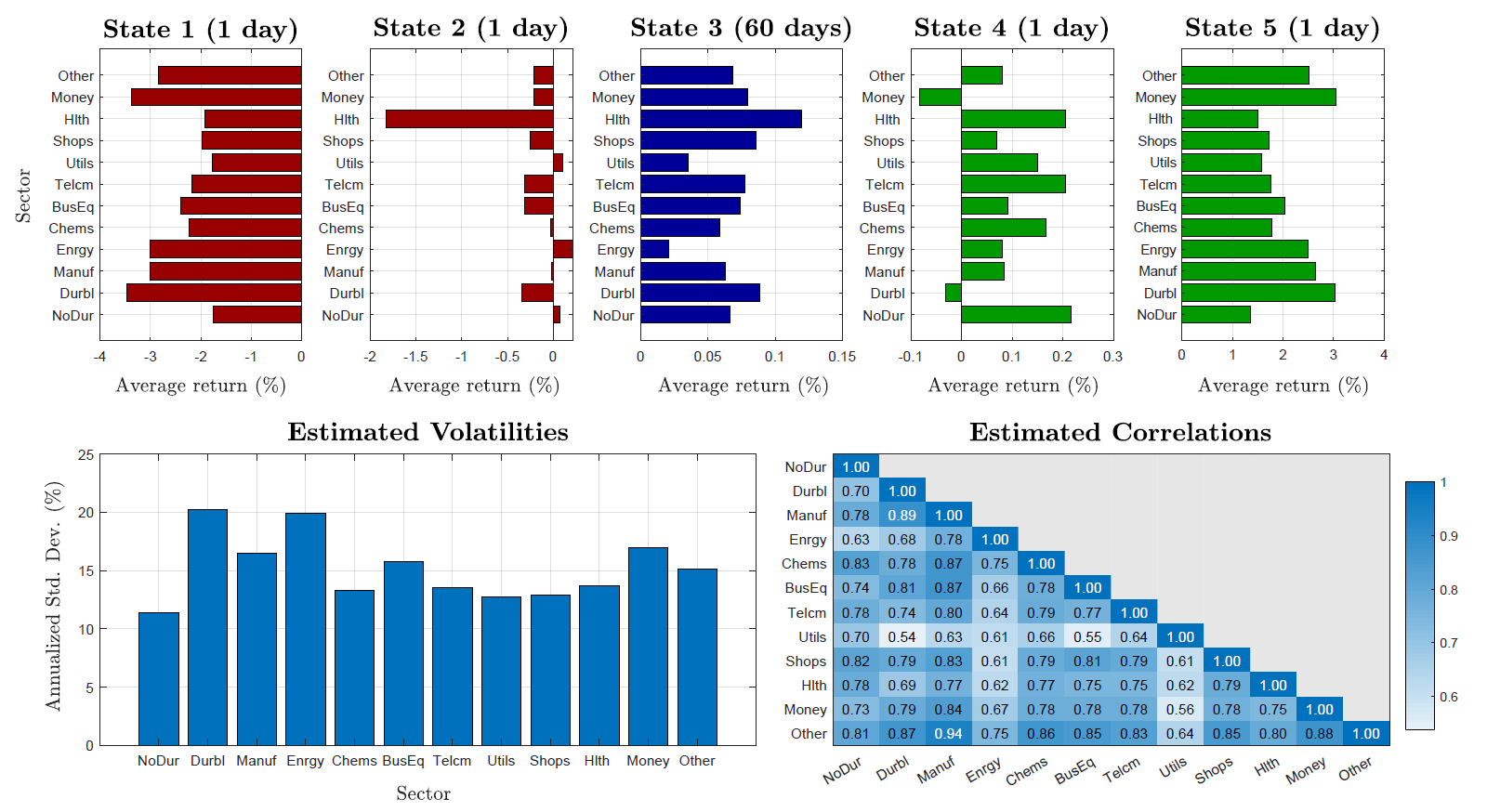}
		\captionsetup{width=.9\linewidth}
		\caption{HMM market model parameter estimation for $m = 5$ states. Top row: estimated growth rates for each asset sorted from left to right sorted by average growth rate (red/blue/green bars indicate ``bad''/``normal''/``good'' states); expected sojourn time given in brackets. Bottom left: estimated volatilities for each asset. Bottom right: estimated correlation matrix.
		}
		\label{fig:HMMparameterEstimation}
	\end{figure}
\end{landscape}

\begin{enumerate}
	\item With an odd number of states, the middle state (sorted by average daily growth rate) has small positive expected return (between 2 and 12 basis points) for all assets. This is close to the overall sample mean and indicates that this is the ``normal'' state of the  market.
	\item The expected sojourn time (time spent in the state once the Markov chain enters it) for the normal state is around 60 days. The states on either side of the middle state are visited for short periods at a time (approximately 1 day on average) and involve mostly positive or negative returns for the entire market. As such, they can be viewed as ``good'' or ``bad'' states of the market. This is also supported by the most likely path obtained using the Viterbi algorithm.
	\item Levels of high activity in the Markov chain over certain time periods, including switching between extreme states, suggest a form of volatility clustering and short-term reversal.
	\item Changing the estimation period or varying the number of states in the model to another odd number yields similar results qualitatively. \\
\end{enumerate}

\begin{figure}[h!]
	\centering
	\includegraphics[width=0.85\textwidth]{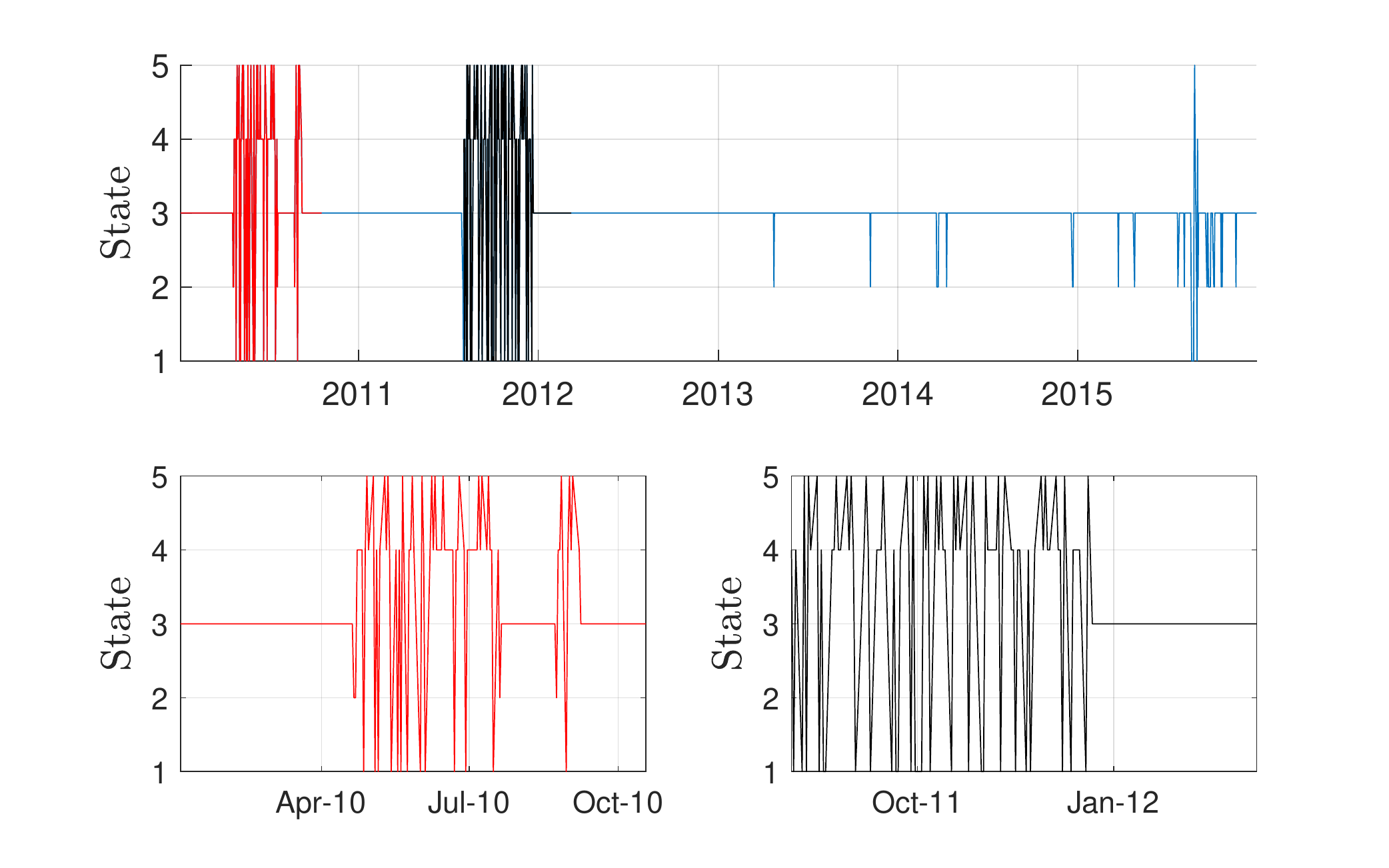}
	\captionsetup{width=.9\linewidth}
	\caption{Most likely path of latent factor according to the Viterbi algorithm (top panel) and Markov chain path in periods of high activity in 2010 (bottom left) and 2011 (bottom right) highlighted by the red and black portions, respectively.
}
	\label{fig:HMMlatentFactorEstimation}
\end{figure}

\vspace{0.5cm}

\begin{table}[h] \footnotesize \centering
	\begin{tabular}{ L{2cm} C{4.5cm} C{3cm}  C{3cm} C{3cm}}
		\toprule[1.5pt]
		\bf{State}    		& \bf{Average Daily Growth Rate Across Industries}		& \bf{Number of Days Spent in State}    & \bf{\%} of Days Spent in State  & \bf{Expected Sojourn Time}   \\
		\midrule	
		\bf{Very Bad}  		& -250 bps &  44   & 3 \%  & 1 day \\					
		\bf{Bad}      	    & -25  bps & 28   & 2 \%  & 1 day \\
		\bf{Normal}  		& 7 bps & 1305 & 86 \% & 60 days  \\
		\bf{Good}  			& 10 bps  & 89   & 6 \%  & 1 day \\
		\bf{Very Good}  	& 210 bps & 44   & 3 \%  & 1 day \\
		\bottomrule[1.5pt]					
	\end{tabular}
	\captionsetup{width=.8\linewidth}
	\caption{HMM estimation results based on applying the EM algorithm to data in the period 2010-2015; estimate of days spent in each state is based on the Viterbi algorithm.}
	\label{tab:viterbiFreq}
\end{table}

\newpage

One difficulty with estimation is that the number of states is assumed known, yet estimating the number of states is a highly non-trivial problem. The likelihood ratio tests for determining the number of states suffers from a lack of identifiability of the alternative under the null hypothesis which causes the classical chi-square theory to fail, see \cite{gassiat2000likelihood} and \cite{gassiat2002likelihood} for more details. To circumvent this issue, we refer to \cite{celeux2008selecting} who compares various likelihood-based approaches to choosing the number of hidden states. They compare the well-known Akaike Information Criteria (AIC) and Bayesian Information Criteria (BIC) selection criteria with the Integrated Complete Likelihood (ICL) approach of \cite{Biernacki2000}, as well as the Odd-Even-Half-Sampling (OEHS) technique developed in their paper. The interested reader is referred to \cite{celeux2008selecting} for a more in-depth discussion of these topics.

To select the number of states we compute the four metrics AIC, BIC, ICL, and OEHS for rolling 5-year periods starting from 1970-1975 and ending at 2012-2017. For each 5-year period, the metrics are computed fixing the number of states from 1 to 11 states. Then the state with the highest metric is chosen for that period. Figure \ref{fig:numberOfStates}  plots the selected number of states for each estimation period and Figure \ref{fig:modelSelection} plots the resulting distributions across all periods.

\begin{figure}
\begin{minipage}[t]{.48\textwidth}
	\centering
	\includegraphics[width=\textwidth]{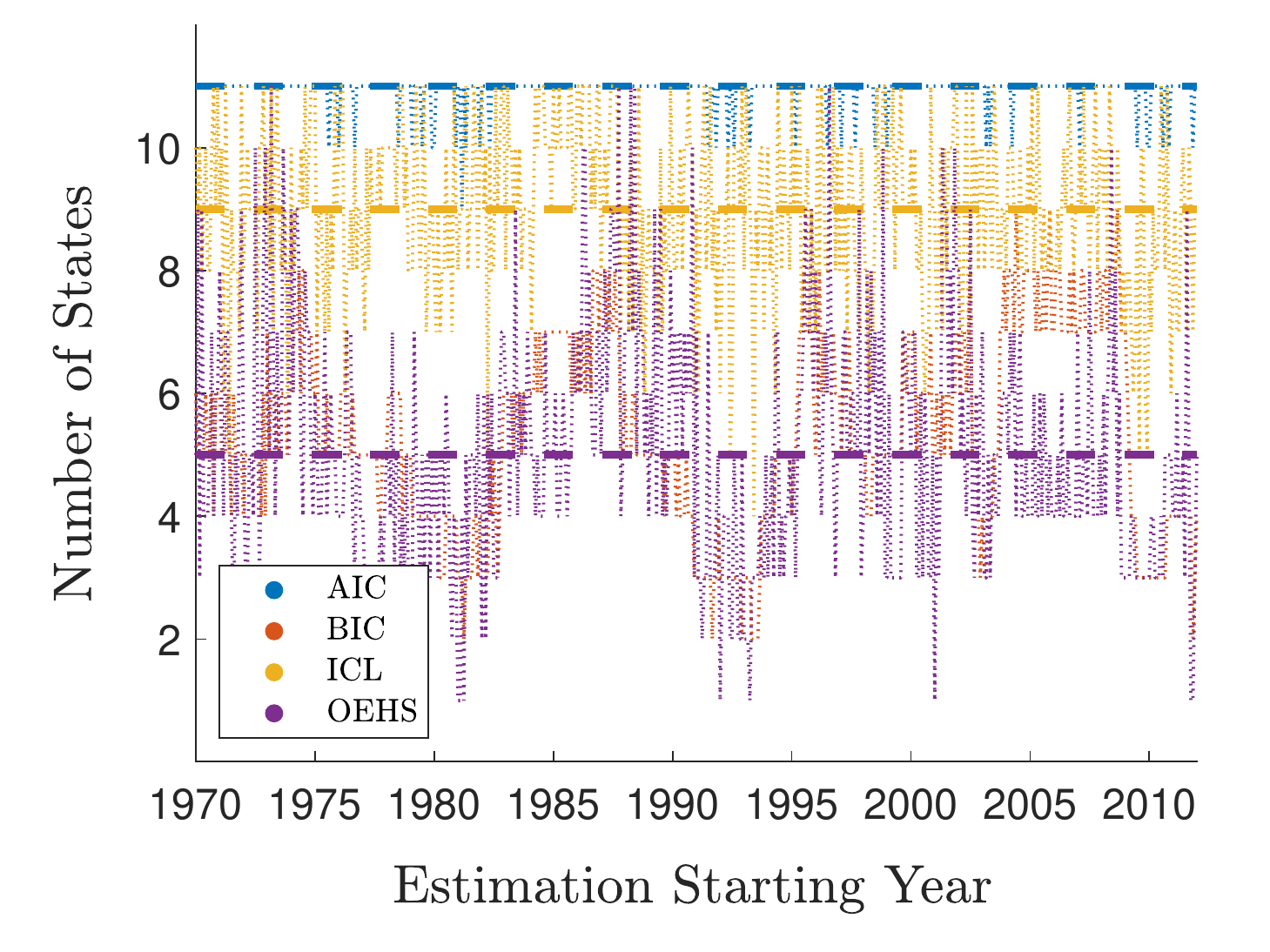}
	\captionsetup{width=.9\linewidth}
	\caption{Preferred number of states according to different model selection metrics for each 5-year estimation period; dashed lines represent the mode of each distribution (most preferred state through time according to the metric).}
	\label{fig:numberOfStates}
\end{minipage}
\begin{minipage}[t]{.48\textwidth}
	\centering
	\includegraphics[width=\textwidth]{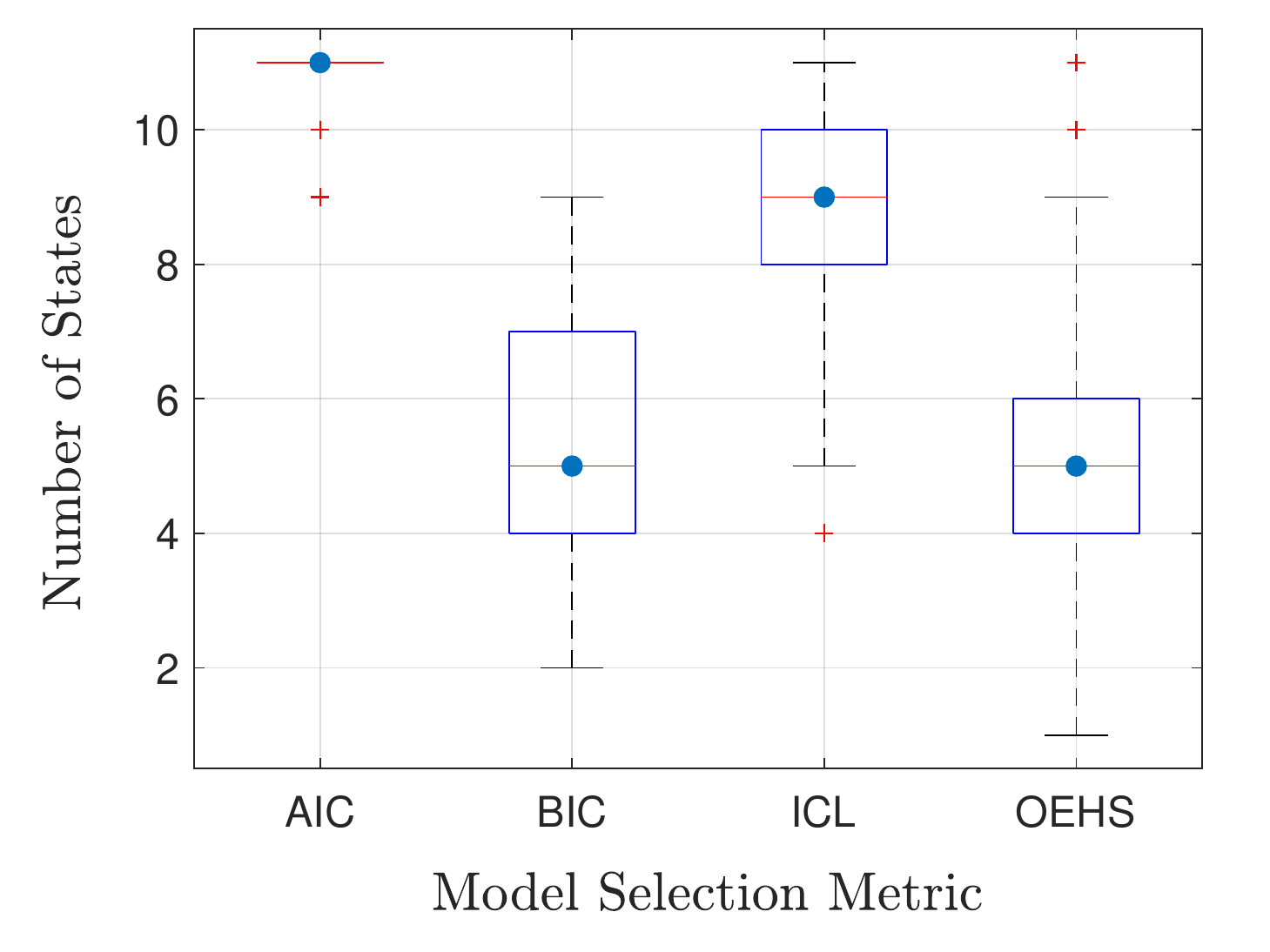}
	\captionsetup{width=.9\linewidth}
	\caption{Boxplots of preferred number of hidden states according to different model selection metrics for rolling 5-year estimation periods; blue dots represent the mode of each distribution (most preferred state according to the metric).}
	\label{fig:modelSelection}
\end{minipage}
\end{figure}

\subsection{Optimal Portfolio Performance}

Here, we present our main backtesting results. The backtests are out-of-sample: they use the previous 5-year period for parameter estimation and implement the optimal portfolio over the subsequent year. This procedure is repeated on a monthly rolling basis, i.e. estimation using Jan-1970 to Dec-1974 data and investment over Jan-1975 to Dec-1975 data for the first backtest, then estimation using Feb-1970 to Jan-1975 data and investment over Feb-1975 to Jan-1976 data for the second backtest, etc. This is repeated for all available data for a total of 505 out-of-sample backtests over the period where data is available.

For parameter estimation, we select the number of hidden states given by the ICL model selection metric (we also assess the robustness of this choice in a later section). For implementation (investment), we compute the optimal portfolio with penalties $\mOmega = \mSigma$ and $\mQ = \mI$ (minimizing active risk and tilting towards the equal-weight portfolio) and use the equal-weight portfolio as the outperformance and tracking benchmark. As is well known, this portfolio outperforms the market portfolio and thus constitutes a more difficult benchmark to outperform. The subjective preference parameters $\vzeta$ are determined differently for each backtest by searching for a $\vzeta$ that arrives as close as possible to the target active risk of 5\%. More specifically, we use the previous 5-year data, which was used for model parameter estimation, along with the estimated model parameters over the same period to search for the $\vzeta$ vector that results in the portfolio with an active risk level closest to the desired level. This $\vzeta$ vector is then fixed and applied in the out-of-sample portion of the backtest. This procedure does not guarantee the active risk over the subsequent 1-year investment period matches the target, but it does provide a rationale for choosing the preference parameters based on the information the investor would have at the time of investing.

In summary, our backtests consist of the following steps:
\begin{enumerate}
	\item Estimate model parameters using 5 years of data with number of states selected by ICL.
	\item Calibrate preference parameters using 5-year in-sample investment.
	\item Implement portfolio (out-of-sample) over subsequent year.
	\item Roll over to next backtest and repeat steps 1-3.
\end{enumerate}	
As a benchmark, we implement the simple GBM model discussed in \cite{alaradi2018} using the analogous out-of sample manner described above (without selecting the number of states, as there are none in this benchmark model). 

Denoting the vector of asset returns between $t$ and $t+1$ by $\boldsymbol{R}_{t,t+1}$, the main comparison metrics (for backtest $i$) are
\begin{align*}
	\text{Outperformance: } & ~~~~ OP_i = \frac{Z^{\vpi^*}_T}{Z^{\veta}_T} - 1\,,
	\\
	\text{Active return: } &  ~~~~ \frac{1}{\Delta t} \frac{1}{T} \sum_{t=1}^T (\vpi^*_t - \veta_t)^\intercal ~\boldsymbol{R}_{t,t+1}\,,
	\\
	\text{Active risk: } &  ~~~~ \sqrt {\frac{1}{\Delta t} \frac{1}{T} \sum_{t=1}^T \left( (\vpi^*_t - \veta_t)^\intercal ~\boldsymbol{R}_{t,t+1} - \mbox{Active Return}\right)^2}\,,
	\\
	\text{Gain-loss ratio: } & ~~~~ \frac{\EE[OP ~|~ OP > 0]}{\EE[OP ~|~ OP < 0]}  \qquad  \text{(averages taken across all backtests)\,. }
\end{align*}

The results are provided in Figures \ref{fig:backtestPlots} and \ref{fig:outperformanceHistogram}. The figures clearly show the optimal strategy under the HMM approach yields better performance relative to its GBM counterpart. The average outperformance across all backtests is 2\% compared to 0.5\% for the GBM model with 70\% of backtests leading to outperformance (compared to only 50\% for the GBM approach) and approximately twice the gain-loss ratio (2.2 comapred to 1.2). In addition, there appears to be an improvement in the level of active risk (2.2\% on average compared to 4.9\%) as well as downside risk (10\% worst-case underperformance compared to 18\%).

\subsection{Robustness Check}

In this section, we test the robustness of the optimal portfolio to model misspecification by considering alternative metrics to choose the number of hidden states. In particular, we repeat the same procedure as in the previous section to conduct out-of-sample backtests using the AIC, BIC and OEHS to determine the number of hidden states. Then, with the new model selection criteria, we average the performance metrics across all the historical backtests. Figure \ref{fig:robustnessPlots} shows estimated densities of outperformance in all backtests for each HMM selection metric along with the proportion of backtests with positive outperformance, the worst-case underperformance and the average gain-loss ratio across all backtests. The plots show consistent improvement in portfolio performance across the various metrics regardless of the model selection criteria. There are, however, a few backtests associated with the OEHS criteria (all involving estimation during 2005-2010) with extraordinarily high outperformance figures. Caution needs to be taken as it is possible for similarly large underperformance figures to occur.

\begin{landscape}
\begin{figure}[h!]
	\centering
	\vspace{0.5cm}
	\includegraphics[width=\linewidth]{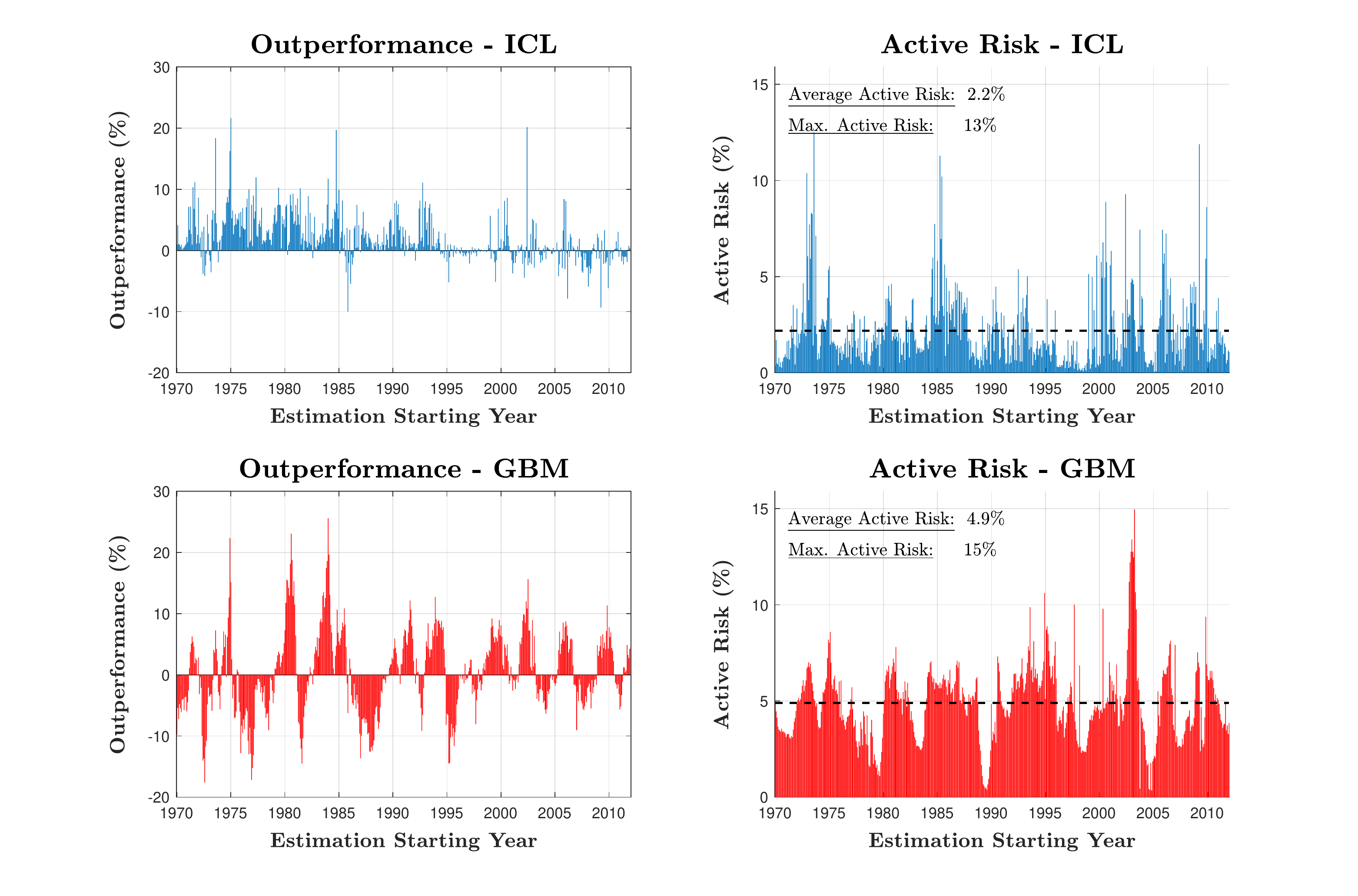}
	\captionsetup{width=.9\linewidth}
	\caption{Outperformance and active risk in each backtest using the HMM approach (top panels) and GBM model (bottom panels).}
	\label{fig:backtestPlots}
\end{figure}
\end{landscape}

\clearpage

\begin{figure}[t!]
	\centering
	\includegraphics[width=0.4\textwidth]{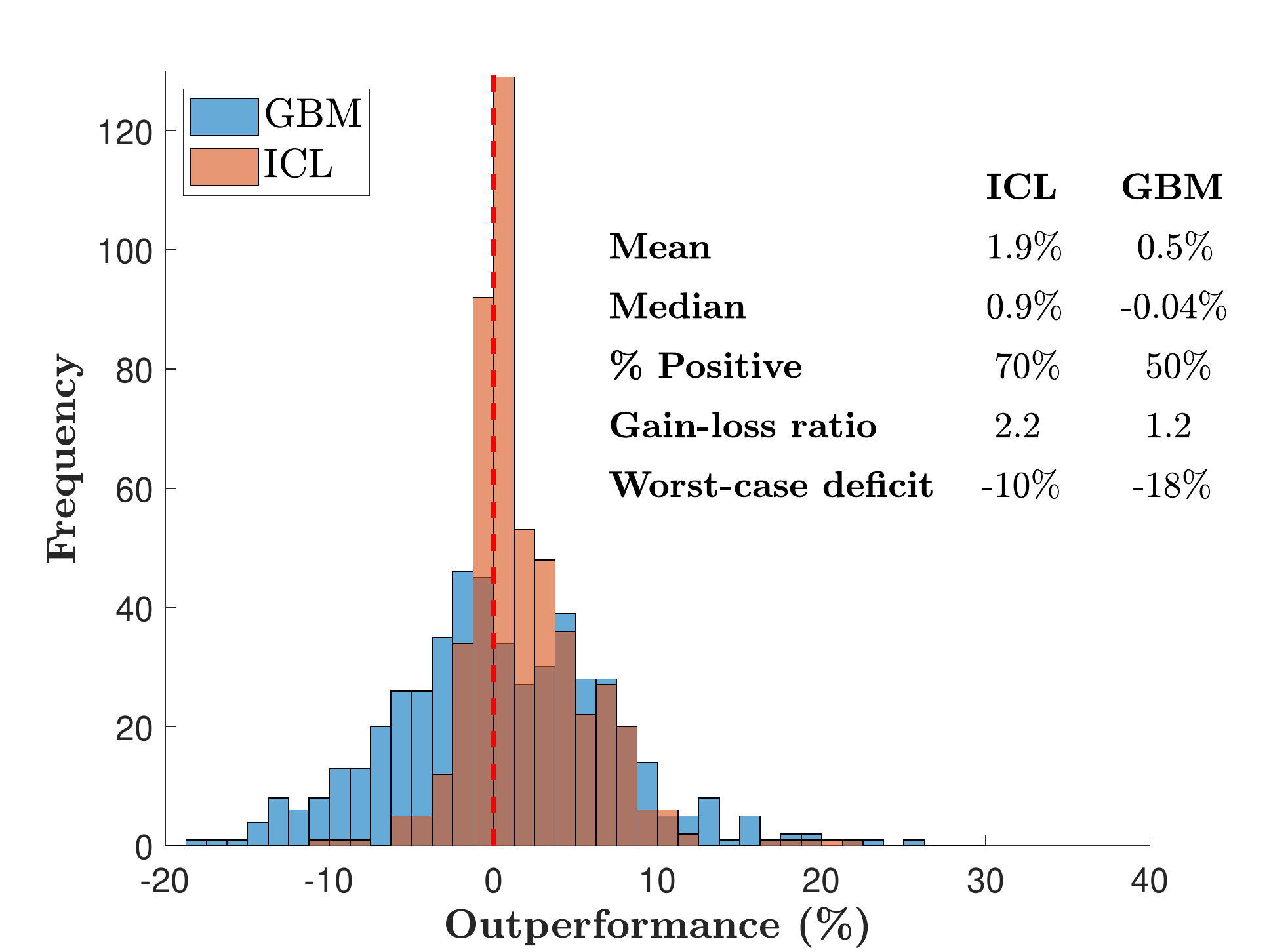}
	\captionsetup{width=.9\linewidth}
	\caption{Histogram of backtest outperformance results using HMM approach and GBM model. 
}
	\label{fig:outperformanceHistogram}
\end{figure}

\begin{figure}[t!]
	\centering
	\includegraphics[width=0.7\textwidth]{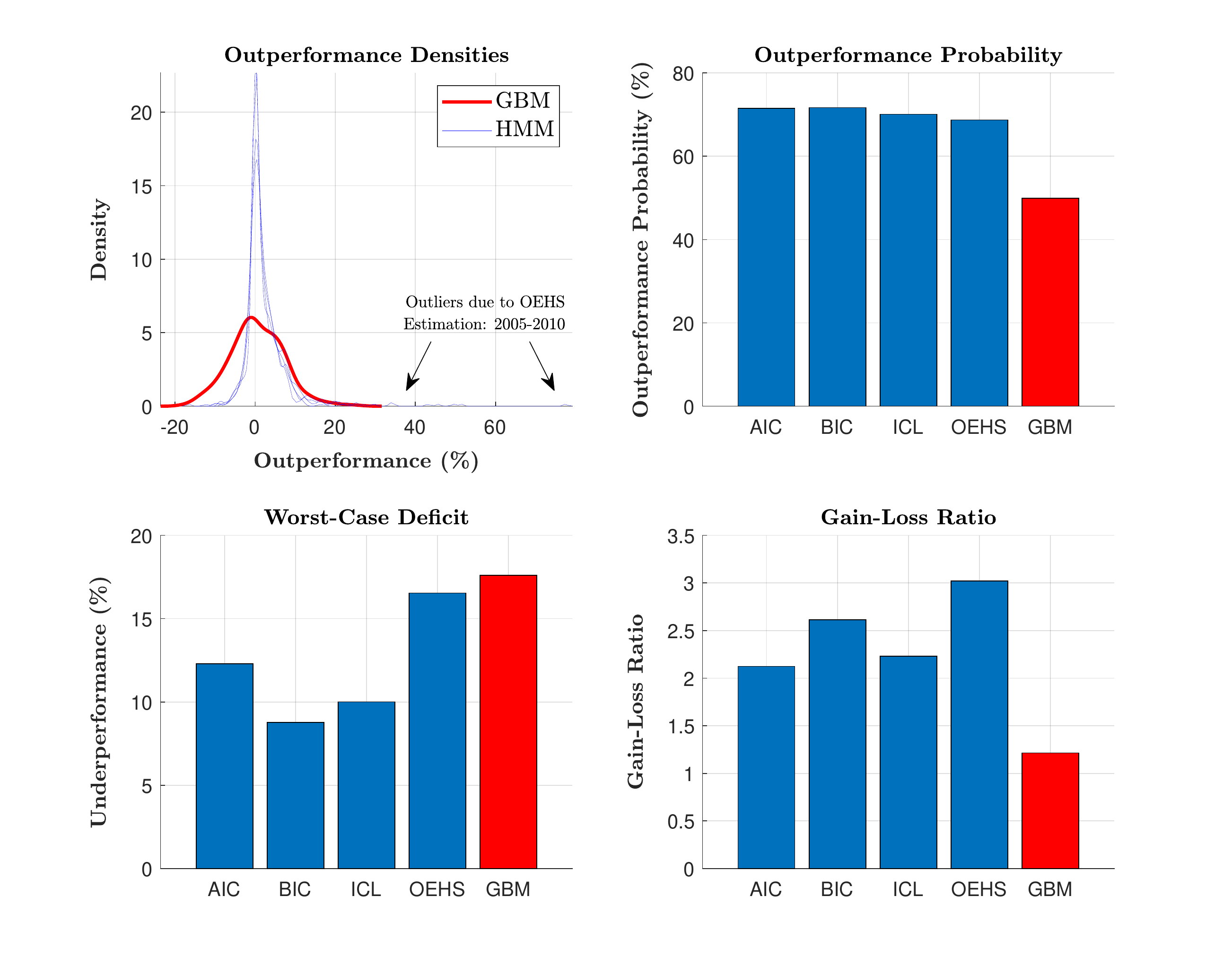}
	\captionsetup{width=.9\linewidth}
	\caption{Statistics from the backtest performance using various model selection criteria; densities of outperformance (top left), probability of outperfoemance (top right), worst-case underperformance (bottom left) and average gain-loss ratio (bottom right). Red bars/lines correspond to GBM/1-state HMM; blue bars/lines correspond to the full HMM implemented using the various model selection metrics.
}
	\label{fig:robustnessPlots}
\end{figure}

\section{Conclusions}

In this paper, we solve the problem of outperformance and tracking faced by an active manager and obtain the optimal portfolio in a fairly general setting. This generalization is achievable using convex analysis techniques. We apply these results to a market model with a latent factor driving asset growth rates with the goal of achieving improved results over the simple GBM model with constant parameters. The out-of-sample backtest results show a marked improvement in outperformance level and probability and the results are robust with respect to model specification.

A number of directions can be taken for future research. It would be interesting, though possibly quite challenging, to consider a model combining rank-based models and HMMs to leverage the stability and estimation qualities of the former approach. Another interesting direction would be to approach the problem using trading rates rather than weight vectors as the investor's control variable. Liquidity risk can be incorporated by penalizing trading rates. Furthermore, it would be interesting to incorporate short-selling constraints into the problem.

\begin{appendices}

\small

\section{Proofs for Section \ref{sec:control} }
\subsection{Proof of Lemma \ref{lemma:innovations}} \label{proof:innovations}
Let $\vMHat_t = \left( \MHat^1_t, ..., \MHat^n_t \right)$. Next,
\begin{align*}
\MHat^i_t ~&=~ \log X^i_t - \int_0^t \gammaHat^i_s ~ds
~=~ \int_0^t \left(\gamma^i_s - \gammaHat^i_s \right) ds  + M^i_t\,.
\end{align*}
Under Assumption \ref{asmp:growthMtg}(b),  $\MHat^i \in \LL^1$ as it is the sum of a bounded term and an integrable term.

Under Assumption \ref{asmp:growthMtg}(a), we show that $\MHat^i \in \LL^2$ by considering
\begin{equation}
 \left( \MHat^i_t \right)^2 ~=~ \left( \int_0^t \left(\gamma^i_s - \gammaHat^i_s \right) ds \right)^2 + \left( M^i_t \right)^2 + 2 M^i_t  \left( \int_0^t \left(\gamma^i_s - \gammaHat^i_s \right) ds \right)\,.
 \label{eqn:Mhat-squared}
\end{equation}
As $\gamma^i,\gammaHat^i\in\LL^2$, their difference is also an $\LL^2$ process and we have
\[ \EE \left[ \left( \int_0^t \left(\gamma^i_s - \gammaHat^i_s \right) ds\right)^2\right] < \infty \,. \]
The second and third terms of \eqref{eqn:Mhat-squared} have finite expectation for all $t$ (up to a zero measure set) as $M^i \in \LL^2$ and by the Cauchy-Schwarz inequality. Combining all of these facts we can conclude that $\vMHat \in \LL^2$.

To show the martingale property, firstly
\begin{align*}
\widehat{\vgamma}_t &= \EE \left[ \vgamma_t ~\middle|~ \Ft_t \right]
\\
\implies \quad \EE \left[ \widehat{\vgamma}_t ~\middle|~ \Ft_t \right] &= \EE \left[ \EE \left[ \vgamma_t ~\middle|~ \Ft_t \right] ~\middle|~ \Ft_t \right] = \EE \left[ \vgamma_t ~\middle|~ \Ft_t \right]
\\
\implies \quad \EE \big[ \left( \vgamma_t - \widehat{\vgamma}_t \right) ~\big|~ \Ft_t  \big] &= 0\,.
\end{align*}
Secondly, consider the conditional expectation of the increment of the innovations process:
\begin{align*}
\EE \left[ \MHat^i_u -  \MHat^i_t ~\middle|~ \Ft_t \right] &= \EE \left[ \int_t^u \left(\gamma^i_s - \gammaHat^i_s \right) ds ~\middle|~ \Ft_t \right] + \redunderbrace{\EE \left[ M^i_u - M^i_t ~\middle|~ \Ft_t \right]}{\text{\scriptsize $=0$, as $M^i$ is a martingale}}
\\
&=  \int_t^u \EE \left[ \left(\gamma^i_s - \gammaHat^i_s \right) ~\middle|~ \Ft_t \right] ~ds
\\
&=  \int_t^u \EE \left[ \EE \left[ \left(\gamma^i_s - \gammaHat^i_s \right) ~\middle|~ \Ft_s \right] ~\middle|~ \Ft_t \right] ~ds
\\
& = 0
\end{align*}
The last two steps follow from the fact that the integrand is in $\LL^2$ which allows us to invoke Fubini's theorem and using the law of iterated expectations conditional on $\sigma$-algebra (Theorem 5.1.6 of \cite{durrett2010probability}). \hfill $\blacksquare$ \\

\subsection{Proof of Proposition \ref{prop:concave}} \label{proof:concave}
We need to show that for $c \in (0,1)$ and $\vpi, \widetilde{\vpi} \in \As$ the functional $H$ satisfies:
\[ H \left( c \vpi + (1-c) \widetilde{\vpi} \right) - c H \left(\vpi \right) -(1-c)H \left(\widetilde{\vpi} \right) > 0\,. \] \\
From the definition of $H$
\begin{align*}
H \left( c \vpi + (1-c) \widetilde{\vpi} \right) &- c H\left(\vpi \right) -(1-c)H\left(\widetilde{\vpi} \right)
\\
& = \EE \Bigg[ \int_0^T  \Bigg\{ - \tfrac{1}{2} \left(c \vpi_t + (1-c) \widetilde{\vpi}_t \right)^\intercal \mA_t \left(c \vpi_t + (1-c) \widetilde{\vpi}_t\right)
\\
& \hspace{2cm} + \frac{c}{2} \left(\vpi_t \right)^\intercal \mA_t \vpi_t + \frac{1-c}{2} \left(\widetilde{\vpi}_t \right)^\intercal \mA_t \widetilde{\vpi}_t
\\
& \hspace{2cm}  + \left(c \vpi_t + (1-c) \widetilde{\vpi}_t \right)^\intercal \vB_t - c \left(\vpi_t \right)^\intercal \vB_t - (1-c) \left(\widetilde{\vpi}_t \right)^\intercal \vB_t
\Bigg\}dt \Bigg]
\\
& = \EE \Bigg[ \int_0^T  \Bigg\{- \tfrac{1}{2} \left(c \vpi_t + (1-c) \widetilde{\vpi}_t \right)^\intercal \mA_t \left(c \vpi_t + (1-c) \widetilde{\vpi}_t\right)
\\
& \hspace{2cm} + \frac{c}{2} \left(\vpi_t \right)^\intercal \mA_t \vpi_t + \frac{1-c}{2} \left(\widetilde{\vpi}_t \right)^\intercal \mA_t \widetilde{\vpi}_t \Bigg\} dt \Bigg]
\\
&  > 0
\end{align*}
The last inequality follows as $\mA$ is negative definite and $f(\vx) = -\tfrac{1}{2} \vx^\intercal \mA \vx$ is concave in $\vx$.

To see that $H$ is proper note that the integrand in \eqref{eqn:perfCrit5} is finite, as quadratic terms involving $\mA$ and $\mOmega$ are assumed to be bounded and $\vpi$ and $\valphaHat$ are either bounded or in $\LL^2$. Hence $H$ nowhere takes the value $-\infty$ and is not identically equal to $\infty$. Furthermore, as $H$ is continuous in $\vpi$ it is also upper semi-continuous. \hfill $\blacksquare$

\subsection{Proof of Proposition \ref{prop:gateaux}} \label{proof:gateaux}
We follow the notation in \cite{ekeland1999convex}. The directional derivative of the functional $H$ is defined by:
\[ H' \left(\vpi; \widetilde{\vpi}\right) = \underset{\epsilon \downarrow 0}{\lim} \frac{H\left(\vpi + \epsilon \widetilde{\vpi}\right) - H\left(\vpi \right)}{\epsilon}\,. \]
The numerator is given by
\begin{align*}
& H\left(\vpi + \epsilon \widetilde{\vpi}\right) - H\left(\vpi\right) \\
& \hspace{1cm} = \EE \biggl[ \int_0^T \biggl\{  - \tfrac{1}{2} \left(\vpi_t + \epsilon \widetilde{\vpi}_t \right)^\intercal \mA_t \left(\vpi_t + \epsilon \widetilde{\vpi}_t \right) + \tfrac{1}{2} \left(\vpi_t\right)^\intercal \mA_t \vpi_t + \epsilon \left(\widetilde{\vpi}_t\right)^\intercal \vB_t \biggr\}~dt \biggr]
\\
& \hspace{1cm} = \EE \biggl[ \int_0^T \biggl\{ - \tfrac{1}{2} \left(\vpi_t\right)^\intercal \mA_t \vpi_t - \epsilon  \left(\vpi_t\right)^\intercal \mA_t \widetilde{\vpi}_t - \tfrac{\epsilon^2}{2}  \left(\widetilde{\vpi}_t\right)^\intercal \mA_t \widetilde{\vpi}_t
\\
& \hspace{3cm} + \tfrac{1}{2} \left(\vpi_t\right)^\intercal \mA_t \vpi_t + \epsilon  \left(\widetilde{\vpi}_t\right)^\intercal \vB_t \biggr\}
~dt \biggr]
\\
& \hspace{1cm} = \EE \biggl[ \int_0^T  \epsilon  \left[ \biggl\{ -\left(\vpi_t\right)^\intercal \mA_t \widetilde{\vpi}_t + \left(\widetilde{\vpi}_t\right)^\intercal \vB_t \right] - \frac{\epsilon^2}{2}  \left(\widetilde{\vpi}_t\right)^\intercal \mA_t \widetilde{\vpi}_t ~\biggr\}dt \biggr]
\end{align*}
Dividing by $\epsilon$ and taking the limit as $\epsilon \downarrow 0$ we have
\begin{align*}
H'\left(\vpi; \widetilde{\vpi}\right)  
&= \EE \biggl[ \int_0^T  \left(\widetilde{\vpi}_t\right)^\intercal \left[ -\mA_t \vpi_t + \vB_t \right] ~dt \biggr]\,,
\end{align*}
which exists for all $\vpi, \widetilde{\vpi} \in \As$. Furthermore, notice that the expression above is in fact a linear continuous functional $H^* \in \As^*$ given by
\[ \langle \widetilde{\vpi} , H^* \rangle =  H^*(\widetilde{\vpi}) = \EE \biggl[ \int_0^T  \widetilde{\vpi}_t^\intercal \left[ -\mA_t \vpi_t + \vB_t \right] ~dt \biggr] \]
and that $\langle \vpi , H^* \rangle = H'(\vpi; \widetilde{\vpi}) $ for all $\vpi \in \As$. Therefore, $H$ is G\^{a}teaux differentiable everywhere in $\As$ with G\^{a}teaux differential given by:
\[ \left\langle \widetilde{\vpi} , H'\left(\vpi\right) \right\rangle = \EE \biggl[ \int_0^T  \left(\widetilde{\vpi}_t\right)^\intercal \left[ -\mA_t \vpi_t + \vB_t \right] ~dt \biggr] \]
\hfill $\blacksquare$

\subsection{Proof of Theorem \ref{thm:optCont}} \label{proof:optCont}

First we show that, under either condition in Assumption \ref{asmp:growthMtg}, $\As^c$ constitutes a closed convex subset of $\As$. To demonstrate convexity, take two portfolios $\vpi^1,\vpi^2 \in \As^c$ and a constant $\lambda \in [0,1]$. Then the process $\vpi^\lambda = \lambda \vpi^1 + (1-\lambda) \vpi^2 $ is $\FF$-adapted since it is the sum of two $\FF$-adapted processes. Furthermore, $\vpi^\lambda \in \LL^2$ if $\vpi^1, \vpi^2 \in \LL^2$ and $\vpi^\lambda \in \LL^{\infty,M}$ if $\vpi^1, \vpi^2 \in \LL^{\infty,M}$. Finally, $(\vpi_t^\lambda)^\intercal \ones = \left(\lambda \vpi_t^1 + (1-\lambda) \vpi_t^2\right)^\intercal \ones = 1$ for all $t \geq 0$ since the elements of $\vpi^1$ and $\vpi^2$ sum to 1 for all $t \geq 0$. Therefore, $\vpi^\lambda \in \As^c$ and hence $\As^c$ is convex. 

Next, we demonstrate that $\As^2$ and $\As^\infty$ are closed subsets of $\As$ under the $\LL^2$-norm. To this end, take a sequence of functions $\{f^k\}_{k \in \NN} \in \As^c$ s.t. $\| f^k - f \|_{\LL^2} \rightarrow 0$ as $k \rightarrow \infty$. As each $f^k \in \As^c$ we also have that $f_k^\intercal \ones = 1$ for each $k$ and we can write
\begin{align*}
\big\| f^\intercal \ones - 1 \big\|^2_{\LL^2}&= \EE \left[\int_0^T \big\| (f_t - f_t^k)^\intercal \ones \big\|^2 ~dt \right]
\\
& \leq \EE \left[\int_0^T \big\| f_t - f_t^k \big\|^2  ~ \big\| \ones \big\|^2 ~dt \right] \qquad {\footnotesize \text{(by the Cauchy-Schwarz inequality)}} 
\\
&= n ~ \big\| f - f^k \big\|^2_{\LL^2} 
\\ 
& \underset{k \rightarrow \infty }{\longrightarrow} 0 \,.
\end{align*}
This implies that $f^\intercal \ones = 1$ for $t$-a.e. $\PP$-a.s. To complete the proof that the subsets are closed we need to show that $f \in \LL^2$ if $\{f^k\}_{k \in \NN} \in \As^2$ or $f \in \LL^\infty$ if $\{f^k\}_{k \in \NN} \in \As^\infty$. Here, we treat the two cases $\As^2$ and $\As^\infty$ separately. If the sequence of functions $\{f^k\}_{k \in \NN} \in \As^2$ then $f \in \LL^2$ due to the triangle inequality:
\[ \| f \|_{\LL^2} \leq \redunderbrace{\| f - f^k \|_{\LL^2}}{$\rightarrow 0$} + \redunderbrace{\| f^	k \|_{\LL^2}}{$< \infty, ~ \forall k$} \]  


If the sequence of functions $\{f^k\}_{k \in \NN} \in \As^\infty$ then we need to show that $f \in \LL^{\infty,M}$, which we will achieve by contradiction.  
%
%
%
%
%
Define the set $A = \{\omega \in \Omega: \|f_t(\omega)\|_\infty > M \} $ and assume that this set has positive measure. Next, we write the expectation in the $\LL^2$-norm in terms of $A$ and $A^c$
\begin{align*}
||f - f_k||^2_{\LL^2} &=\EE \int_0^T ||f_t(\omega) - f_t^k(\omega)||^2 \,dt 
\\
&= \int_{\omega \in \Omega} \int_0^T ||f_t(\omega) - f_t^k(\omega)||^2 \,dt \,d\PP(\omega)
\\
& \geq \int_{\omega \in  A} \int_0^T ||f_t(\omega) - f_t^k(\omega)||^2 ~dt ~d\PP(\omega) \,.
\end{align*}
Since the LHS tends to zero in $k$ and RHS is non-negative, we have 
\[\int_{\omega \in A} \int_0^T ||f_t(\omega) - f_t^k(\omega)||^2 \,dt \,d\PP(\omega) ~~ \underset{k \rightarrow \infty}{\longrightarrow} ~~ 0 \,. \]
Using the following inequality for $p$-norms and the $\infty$-norm on $\RR^n$
\[ \| f_t - f_t^k \|_\infty \leq \| f_t - f_t^k \|_p \leq n^{1/p} \| f_t - f_t^k \|_\infty \, , \]
it follows that for $p=2$
\[ \| f_t - f_t^k \|_\infty \leq \| f_t - f_t^k \|_\infty^2 \leq \| f_t - f_t^k \|^2 \, . \]
Now, by Minkowsi's inequality $\| f_t \|_\infty \leq \| f_t - f_t^k \|_\infty + \| f_t^k \|_\infty $ we have that 
\[ \| f_t \|_\infty - \| f_t^k \|_\infty \leq \| f_t - f_t^k \|^2  \,. \] 
Integrating over $[0,T]$ and $A$ we obtain
\[ \int_{\omega \in  A} \int_0^T \left(\| f_t(\omega) \|_\infty - \| f_t^k(\omega) \|_\infty\right) \,dt \,d\PP(\omega) ~\leq~ \int_{\omega \in  A} \int_0^T \| f_t(\omega) - f_t^k(\omega) \|^2 \,dt \,d\PP(\omega) \, \] 
and as $f_t^k \in \LL^{\infty,M}$ we have $\| f_t^k(\omega) \|_\infty \leq M$ for all $t$ and we can write
\[ \int_{\omega \in  A} \int_0^T  \left(\|f_t(\omega) \|_\infty - M\right) \,dt \,d\PP(\omega) ~\leq~ \int_{\omega \in  A} \int_0^T \| f_t(\omega) - f_t^k(\omega) \|^2 \,dt \,d\PP(\omega) \,. \] 
However, as the RHS tends to zero and the integrand is strictly positive on $A$ we arrive at a contradiction. Thus, $\PP(A) = 0$ which implies that $f \in \LL^{\infty,M}$.


 
To show the candidate optimal portfolio $\vpi^*$ given by \eqref{eqn:opt-pi-A-and-B} is in fact optimal, it suffices to show that it is an element of $\As^c$ and that the G\^{a}teaux derivative of $H$ vanishes at $\vpi^*$, i.e. that it satisfies $\langle \vpi - \vpi^* , H'(\vpi^*) \rangle = 0$ for all $\vpi \in \As^c$. 

To show that $\vpi^* \in \As_c$, firstly
\begin{align*}
\vpi^*_t &= \mA^{-1}_t  \left[ \frac{1 - \ones^\intercal \mA^{-1}_t \vB_t}{\ones^\intercal \mA^{-1}_t \ones} ~ \ones + \vB_t \right]
\\
&= \left(\frac{1 - \ones^\intercal \mA^{-1}_t \vB_t}{\ones^\intercal \mA^{-1}_t \ones}\right)  \mA^{-1}_t \ones + \mA^{-1}_t \vB_t
\\
&= \left(\frac{1 - \ones^\intercal \mA^{-1}_t  \vB_t}{\ones^\intercal \mA^{-1}_t \ones}\right)  \mA^{-1}_t \ones + \mA^{-1}_t \bigg[ \zeta^0 \valphaHat_t + \zeta^1_t \mOmega_t \veta_t \bigg]
\\
&= \left(\frac{1 - \ones^\intercal \mA^{-1}_t  \vB_t}{\ones^\intercal \mA^{-1}_t \ones}\right)  \mA^{-1}_t \ones + \zeta^1_t  \mA^{-1}_t \mOmega_t \veta_t + \zeta^0 \mA^{-1}_t \valphaHat_t
\\
&= \left(\frac{1 - \ones^\intercal \mA^{-1}_t \vB_t}{\ones^\intercal \mA^{-1}_t \ones}\right) \mA^{-1}_t \ones + \zeta^1_t \mA^{-1}_t \mOmega_t \veta_t
+ \tfrac{1}{2} \zeta^0  \text{diag}\left(\mSigma_t\right) + \zeta^0 \mA^{-1}_t \vgammaHat_t
\end{align*}
Recall that we have $\mA_t =  \zeta^0 \mSigma_t + \zeta^1_t \mOmega_t + \zeta^2_t \mQ_t$. 

Secondly, notice the individual terms of $\mA^{-1}$ and $\mSigma$ as well as quadratic forms involving $\mA^{-1}$ and $\mOmega$ are bounded (see proof of Proposition 2 in \cite{alaradi2018}). Therefore, the first three terms of the above sum are bounded. 

The remaining term, $\zeta^0 \mA^{-1}_t \vgammaHat_t$, is a linear combination of the projected growth rates $\gammaHat_i$. From the proof of Proposition \ref{prop:projection}, these (and hence $\vpi^*$) are bounded when the growth rates $\vgamma$ are bounded and belong to $\LL^2$ when $\vgamma \in \LL^2$. Finally, as it is easy to verify that the $(\vpi^*_t)^\intercal\ones = 1$ it follows that $\vpi^* \in \As^c$.

To show the optimality condition we note that:
\begin{align*}
\langle \vpi - \vpi^* , H'(\vpi^*) \rangle &= \EE \biggl[ \int_0^T  (\vpi_t - \vpi^*_t)^\intercal \left[ -\mA_t \vpi^*_t + \vB_t \right] ~dt \biggr]
\\
&= \EE \biggl[ \int_0^T  (\vpi_t - \vpi^*_t)^\intercal \left[ -\mA_t \mA^{-1}_t \left[ \frac{1 - \ones^\intercal \mA^{-1}_t \vB_t}{\ones^\intercal \mA^{-1}_t \ones} ~ \ones + \vB_t \right] + \vB_t \right] ~dt \biggr]
\\
&= \EE \biggl[ \int_0^T  (\vpi_t - \vpi^*_t)^\intercal \left[ -\frac{1 - \ones^\intercal \mA^{-1}_t \vB_t}{\ones^\intercal \mA^{-1}_t \ones} ~ \ones - \vB_t + \vB_t \right] ~dt \biggr]
\\
&= -\EE \biggl[ \int_0^T  \frac{1 - \ones^\intercal \mA^{-1}_t \vB_t}{\ones^\intercal \mA^{-1}_t \ones} ~ \redunderbrace{(\vpi_t - \vpi^*_t)^\intercal \ones}{$=0 \text{ for } \vpi \in \As^c$}  ~dt \biggr]
\end{align*}
As $\langle \vpi - \vpi^* , H'(\vpi^*) \rangle = 0$ for all $\vpi \in \As^c$, it follows that $\vpi^*$ attains the sup in $\underset{\vpi \in \As_c}{\sup} H(\vpi)$ by Proposition 2.1 in Chapter 2 of \cite{ekeland1999convex}. \hfill $\blacksquare$

\section{Proofs for Section \ref{sec:HMM} }

\subsection{Proof of Lemma \ref{lemma:unnormalizedPosterior}}  \label{proof:unnormalizedPosterior}

Since $\vlambda$ satisfies Novikov's condition, it follows that
\[ \exp\left[ - \int_0^T \vlambda_{u^-}^\intercal ~d\vW_u - \frac{1}{2} \int_0^T \vlambda_{u^-}^\intercal \vlambda_{u^-} du \right]\,, \]
is a valid Radon-Nikodym derivative from $\PP$ to some probability measure which we denote $\widetilde{\PP}$. The measure change from $\widetilde{\PP}$ to $\PP$ is given by the Radon-Nikodym derivative
\[ \frac{d\PP}{d \widetilde{\PP}} = \exp\left[ \int_0^T \vlambda_{u^-}^\intercal ~d\vW_u - \frac{1}{2} \int_0^T \vlambda_{u^-}^\intercal \vlambda_{u^-} du \right]\,. \]
Now define the process $\Upsilon = \{\Upsilon_t\}_{t \geq 0}$ to be the conditional expectation of the stochastic exponential given above, i.e. 
\begin{equation*}
\qquad \Upsilon_t = \EE\left[ \frac{d\PP}{d \widetilde{\PP}} ~\middle|~ \Ft_t \right] = \exp\left[ \int_0^t \vlambda_{u^-}^\intercal ~d\widetilde{\vW}_u - \frac{1}{2} \int_0^t \vlambda_{u^-}^\intercal \vlambda_{u^-} du \right]\,.
\end{equation*} 
The desired result follows by applying the properties of conditional expectations when changing measures:
\begin{align*}
p^j_t ~&=~ \EE\left[ \ind_{\{\Theta_t = j\}} ~\middle|~ \Ft_t  \right]
\\
~&=~ \frac{\EE^{\widetilde{\PP}} \left[ \ind_{\{\Theta_t = j\}} \Upsilon_t ~\middle|~ \Ft_t  \right]}{\EE^{\widetilde{\PP}} \left[ \Upsilon_t ~\middle|~ \Ft_t  \right]}
\\
~&=~ \frac{\EE^{\widetilde{\PP}} \left[ \ind_{\{\Theta_t = j\}} \Upsilon_t ~\middle|~ \Ft_t  \right]}{ \sum_{i=1}^m \EE^{\widetilde{\PP}} \left[ \ind_{\{\Theta_t = i\}} \Upsilon_t ~\middle|~ \Ft_t  \right]}
\\
~&=~ \frac{P^j_t}{ \sum_{i=1}^m P^i_t }
\end{align*}
where $P^j=(P_t^j)_{t \geq 0}$ is defined as $P_t^j =\EE^{\widetilde{\PP}} \left[ \ind_{\{\Theta_t = j\}} \Upsilon_t ~\middle|~ \Ft_t  \right]$.

\subsection{Proof of Theorem \ref{thm:filter}} \label{proof:filter1}
The process $\Upsilon = \{\Upsilon_t\}_{t \geq 0}$ satisfies the SDE
\begin{equation*}
d \Upsilon_t = \Upsilon_{t^-}\,\vlambda_{t^-}^\intercal ~d\widetilde{\vW}_t\,.
\end{equation*}
Furthermore, Girsanov's theorem implies the process $\vWtilde$ defined by
\[ 
\widetilde{\vW}_t = \int_0^t \vlambda_s ~ds + \vW_t\,,
\]
is a standard $\widetilde{\PP}$-Wiener process. Let $\vlambda_t = \vxi^\intercal \mSigma^{-1} \vgamma_t$ and note that it satisfies Novikov's condition as it is bounded. Substituting $\vlambda$ into the asset price dynamics in \eqref{eqn:logPrice} implies
\begin{align*}
d \log \vX_t &= \vgamma_{t^-} ~dt + \vxi \left( -\vlambda_{t^-} dt + d\widetilde{\vW}_t \right)
\\
&= \left(\vgamma_{t^-} - \vxi \vlambda_{t^-} \right) dt + \vxi ~d\widetilde{\vW}_t
\\
&= \vxi ~d\widetilde{\vW}_t\,.
\end{align*}
Note that $\vlambda$ is bounded and $\widetilde{\vW}_t$ is $\FF$-adapted as $\vxi$ is constant. 

Denote the indicator  $\ind^j_t \coloneqq \ind_{\{\Theta_t = j\}}$. It satisfies the SDE
\begin{align*}
\qquad d \ind^j_t = \sum_{i=1}^m \ind^i_{t^-} \mG_{ji} ~dt + d \Ms^j_t
\end{align*}
where $\Ms^j$ is a square-integrable $\FF$-adapted, $\widetilde{\PP}$-martingale (see \cite{rogers1994}). Applying the product rule for semimartingales, the process $\left( \ind^j_t  \Upsilon_t \right)_{t \geq 0}$ satisfies the SDE
\begin{align*}
d \left(\ind^j_t  \Upsilon_t \right) &= \ind^j_{t^-} d \Upsilon_t + \Upsilon_{t^-} d \ind^j_t + {\color{red} \underbrace{\color{black} d[\ind^j, \Upsilon ]_t}_{ { \color{red} =0 \text{ } }} }
\\
&= \sum_{i=1}^m \ind^i_{t^-} \Upsilon_{t^-} \mG_{ji} ~dt + \Upsilon_{t^-} d \Ms^j_t + \ind^j_{t^-} \Upsilon_{t^-} \vlambda_{t^-}^\intercal ~d\widetilde{\vW}_t
\end{align*}
The covariation term in line 1 vanishes, as $\Upsilon_t$ is continuous and $\ind^j_t$ is a pure jump processes. From the definition of $P^j_t$ above we have:
\begin{align*}
P^j_t &= \EE^{\widetilde{\PP}} \left[  \ind^j_t \Upsilon_t ~\middle|~ \Ft_t  \right]
\\
&= \EE^{\widetilde{\PP}} \biggl[  \ind^j_0 \Upsilon_0 + \int_0^t \sum_{i=1}^m \ind^i_{u^-} \Upsilon_{u^-} \mG_{ji} ~du
\\
& \qquad\qquad\quad + \int_0^t \ind^j_{u^-} \Upsilon_{u^-} \vlambda_{u^-}^\intercal ~d\widetilde{\vW}_u  + \int_0^t \Upsilon_{u^-} ~d \mathcal{M}^j_u ~\bigg|~ \Ft_t \biggr]
\\
&= { \EE^{\widetilde{\PP}} \biggl[  \ind^j_0 \Upsilon_0 + \int_0^t \sum_{i=1}^m  \ind^i_{u^-} \Upsilon_{u^-} \mG_{ji} ~du + \int_0^t \ind^j_{u^-} \Upsilon_{u^-}  \vgamma_{u^-}^\intercal \mSigma^{-1} \vxi ~d\widetilde{\vW}_u ~\bigg|~ \Ft_t \biggr] }
\end{align*}
where the last term vanishes as $\Upsilon$ is square-integrable and $\Ms$ is a $\widetilde{\PP}$-martingale.

Using $\ind^j_t \vgamma_t = \ind^j_t \vgamma^{(j)}_t$, we can write $P^j_t$ as
\begin{align*}
P^j_t ~&=~ \EE^{\widetilde{\PP}} \biggl[  \ind^j_0 \Upsilon_0 + \int_0^t \sum_{i=1}^m \ind^i_{u^-} \Upsilon_{u^-} \mG_{ji} ~du + \int_0^t \ind^j_{u^-} \Upsilon_{u^-} \left(\vlambda^{(j)}_u\right)^\intercal d\widetilde{\vW}_u ~\bigg|~ \Ft_t \biggr]
\end{align*}
where $\vlambda^{(j)}_t = \vxi^\intercal \mSigma^{-1} \vgamma^{(j)}_t$ is $\FF$-adapted. 

Next, as the integrands are square-integrable, expectation and integration can be interchanged, and we have
\begin{align*}
P^j_t ~&=~ \EE^{\widetilde{\PP}} \biggl[  \ind^j_0 \Upsilon_0 ~\bigg|~ \Ft_t \biggr] + \EE^{\widetilde{\PP}} \biggl[  \int_0^t \sum_{i=1}^m \ind^i_{u^-} \Upsilon_{u^-} \mG_{ji} ~du ~\bigg|~ \Ft_t \biggr]
\\
& \qquad\qquad\qquad\qquad + \EE^{\widetilde{\PP}} \biggl[  \int_0^t \ind^j_{u^-} \Upsilon_{u^-} \left(\vlambda^{(j)}\right)^\intercal ~d\widetilde{\vW}_u ~\bigg|~ \Ft_t \biggr]
\\
~&=~ \EE^{\widetilde{\PP}} \biggl[  \ind^j_0 \Upsilon_0 ~\bigg|~ \Ft_0 \biggr] +   \int_0^t \sum_{i=1}^m \EE^{\widetilde{\PP}} \biggl[ \ind^i_{u^-} \Upsilon_{u^-} ~\bigg|~ \Ft_u \biggr] \mG_{ji} ~du
\\
& \qquad\qquad\qquad\qquad +  \int_0^t \EE^{\widetilde{\PP}} \biggl[ \ind^j_{u^-} \Upsilon_{u^-} ~\bigg|~ \Ft_u \biggr] \left(\vlambda^{(j)}\right)^\intercal  d\widetilde{\vW}_u
\\
~&=~ P^j_0 +   \int_0^t \sum_{i=1}^m P^i_u \mG_{ji} ~du + \int_0^t P^j_u \left(\vlambda^{(j)}\right)^\intercal  d\widetilde{\vW}_u\,.
\end{align*}
Alternatively, in differential form we have
\begin{equation*}
dP^j_t ~=~  \sum_{i=1}^m P^i_{t^-} \mG_{ji} ~dt + P^j_{t^-}\,\left(\vlambda^{(j)}\right)^\intercal ~ d\widetilde{\vW}_t\,,
\end{equation*}
with initial conditions $P^j_0 = p^j_0$. Noting that
\begin{align*}
\left(\vlambda^{(j)}\right)^\intercal d\widetilde{\vW}_t &=  \left(\vgamma^{(j)}\right)^\intercal \mSigma^{-1} \vxi ~ d\widetilde{\vW}_t = \left(\vgamma^{(j)}\right)^\intercal \mSigma^{-1} ~ d \log \vX_t
\end{align*}
completes the proof. \hfill $\blacksquare$

\section{Derivation of Discretized Filter \eqref{eqn:discreteFilter} }
\label{proof:discreteFilter}
We begin by rewriting the SDE system for the filter \eqref{eqn:filterSDE}  in vector notation as follows
\begin{equation}
\mDim{d \vP_t}{m}{1} = \mDim{\mG \vphantom{\vP_t} }{m}{m} \mDim{\vP_t}{m}{1} ~dt + \mDim{\mB_t \vphantom{d \log \vX_t}}{m}{n} \mDim{\vphantom{d \log \vX_t} \mSigma^{-1}}{n}{n} ~\mDim{d \log \vX_t}{n}{1}\,,
\label{eqn:dvPusingB}
\end{equation}
where 
\[
\mDim{\mB_t}{m}{n} =
\left( \begin{array}{c}
\mDim{P^1_t \left( \vgamma^{(1)} \right)^\intercal}{1}{n} \\
\vdots \\
P^m_t \left( \vgamma^{(m)} \right)^\intercal
\end{array}
\right)
=
\left(\begin{array}{ccc}
P^1_t \gamma^{(1)}_1 & \cdots & P^1_t \gamma^{(1)}_n \\
\vdots & \ddots & \vdots \\
P^m_t  \gamma^{(m)}_1 & \cdots & P^m_t  \gamma^{(m)}_n
\end{array}\right)
\]
Now, factor $\vP$ as follows
\[
\vP_t = e^{\mG(t-u)} \left( \begin{array}{c}
P^1_u ~ Y^1_t \\
\vdots \\
P^m_u ~ Y^m_t
\end{array}
\right)\,,
\]
with  $Y^j_u = 1$, for $j=1,\dots,m$, $u\le t$ denotes the start of the discretization interval, and $\{(Y_t^j)_{t\in[u, u+\Delta u]}\}_{j=1,\dots,m}$ are to be determined.
Taking the differential and using \eqref{eqn:dvPusingB}  we find
\[ 
\left( \begin{array}{c}
P^1_u ~ dY^1_t \\
\vdots \\
P^m_u ~ dY^m_t
\end{array}
\right) = e^{-\mG(t-u)} \mB_t \mSigma^{-1} d \log \vX_t\,. 
\]
Substituting in the expression for $\mB_t$ and componentwise dividing out $P^j_u$ we have
\[ \left( \begin{array}{c}
dY^1_t \\
\vdots \\
dY^m_t
\end{array}
\right) = e^{-\mG(t-u)}
\left( \begin{array}{c}
\tfrac{P^1_t}{P^1_u} ~ \left(\vgamma^{(1)}\right)^\intercal \mSigma^{-1} ~ d \log \vX_t \\
\vdots \\
\tfrac{P^m_t}{P^m_u} ~ \left(\vgamma^{(m)}\right)^\intercal \mSigma^{-1} ~ d \log \vX_t
\end{array}
\right)
\]

Taking a left limit approximation with $e^{-\mG(t-u)} \approx \mI$ we have that $Y^j_t = \tfrac{P^j_t}{P^j_u}$, and the SDEs for $Y_t^j$ decouple to give
\[ dY^j_t = Y^j_t  \left(\vgamma^{(j)}\right)^\intercal \mSigma^{-1} ~ d \log \vX_t \qquad \mbox{ for } j = 1,..., m \]
Integrating from $u$ to $u+\Delta u$ we arrive at the discretization scheme in \eqref{eqn:discreteFilter}.

\section{EM Algorithm for HMM Market Model \eqref{eqn:hmm}} \label{sec:EM}

Starting from the HMM market model \eqref{eqn:hmm} we have
\begin{align*}
& d \log \vX_t = \vgamma^{(\Theta_t)} ~dt + \vxi ~d\vW_t
\\
\implies \qquad & \log \vX(t + \Delta t) - \log \vX_t \sim N \left( \int_{t}^{t + \Delta t}  \vgamma^{(\Theta_s)} ~ds, ~ \Delta t ~ \mSigma  \right)
\end{align*}
Assuming Markov switching only occurs at the discrete time points $\mathfrak{T} = \{\Delta t, 2 \Delta t, ..., T = N \Delta t\}$, the integral in the mean of the normal distribution above can be simplified:
\[ 
\log \vX(t + \Delta t) - \log \vX_t \sim N \left( \Delta t ~ \vgamma^{(\Theta_t)}, ~ \Delta t ~ \mSigma  \right) \qquad \mbox{ for } t = 0, \Delta t, ..., T - \Delta t \,.
\]
The parameter set here is:
\begin{itemize}
	\item The collection of mean vectors $\vgamma^{(j)}$ for $j \in \mfM$;
	\item the shared covariance matrix $\mSigma$;
	\item the probability transition matrix $(\mZ_{ij})_{i,j \in \mfM}$; 
	\item the probability distribution of the initial state $\{\mathfrak{p}_j\}_{j \in \mfM}$.
\end{itemize}
	  We now summarize the EM algorithm as applied to this parameter estimation problem - further details can be found in \cite{bishop2006pattern}. The E-step involves finding the posterior distribution of the latent variables given a parameter set. More specifically, this requires us to compute the conditional probability of the Markov chain being in state $k$ at time step $n$, denoted $a(\Theta_{nk})$, and the posterior probability of transitioning from state $j$ at time $n-1$ to state $k$ at time $n$, denoted $b(\Theta_{n-1,j}, \Theta_{nk})$. These quantities are estimated via the forward-backward algorithm described in Figure \ref{fig:EMalgo}.

\begin{figure}[h!]
	{\small
		\begin{tabular}{r}
			\hline\hline
			\hspace{0.85\textwidth}
		\end{tabular}
		\vspace{-1em}
		\begin{enumerate}
			\item \texttt{Compute the Gaussian emission density $\phi^k_n = \phi(\vx_n | \vgamma^{(k)}, \mSigma)$ \\ in each state $k \in \mfM$ and at each time step $n \in \mathfrak{T}$.}
			
			\item \texttt{Compute the alpha function (forward algorithm):}
			\begin{enumerate}
				\item \texttt{Initialize: ${\alpha}^k_{1} = \mathfrak{p}_k \phi^k_{1}$ for $k \in \mfM$}
				\item \texttt{Set $n = 2$}.
				\item \texttt{Update at each time step: $\widetilde{\alpha}^k_{n} = \phi^k_{n} \sum_{j = 1}^m \mZ_{kj}^\intercal \alpha^j_{n-1}$} 
				\item \texttt{Compute the normalization constant: $C_n = \sum_{k = 1}^m \widetilde{\alpha}^k_{n}$}
				\item \texttt{Normalize alpha value: $\alpha^k_n = \frac{\widetilde{\alpha}^k_{n}}{C_n}$} 
				\item \texttt{Set $n\to n+1$, if $n \leq N$ go to step (c).}
			\end{enumerate}
						
			\item \texttt{Compute the beta function (backward algorithm)}:
			\begin{enumerate}
				\item \texttt{Initialize: ${\beta}^k_{N} = 1$ for $k \in \mfM$}
				\item \texttt{Set $n = N-1$}.
				\item \texttt{Update at each time step: $\beta^k_{n} = \frac{1}{C_{n+1}} \sum_{j =1}^m \mZ_{kj} \beta^j_{n+1} \phi^k_{n+1}$} 
				\item \texttt{Set $n\to n+1$, if $n \geq 0$ go to step (c).}
				
			\end{enumerate}
			\item \texttt{Compute $a$ and $b$ for all $n \in \mathfrak{T}$ and $j, k \in \mfM$:}
			\[ a(\Theta_{nk}) = \frac{\alpha_n^k \beta_n^k}{\sum_{j = 1}^m {\alpha_n^j \beta_n^j}}  \qquad \qquad b(\Theta_{n-1,j}, \Theta_{nk}) = \frac{\alpha^j_{n-1} \mZ_{jk} \phi_n^k \beta_n^k}{\sum_{i,l = 1}^{m} \alpha^i_{n-1} \mZ_{il} \phi_n^l \beta_n^l} \]
		\end{enumerate}
		\begin{tabular}{r}
			\hline\hline
			\hspace{0.85\textwidth}
		\end{tabular}
	}
	\vspace{-1em}
	\caption{Forward-backward algorithm for E-step of EM algorithm. \label{fig:EMalgo}}
\end{figure}

The M-step is then applied to update the parameter estimates. Applying the first order conditions to each $\vgamma^{(j)}$ and $\mSigma$ we obtain the update rules
\begin{align*}
\vgamma^{(j)} &= \frac{\sum_{n=1}^{N} a(\Theta_{nj}) \vx_n}{\sum_{n=1}^{N} a(\Theta_{nj})}\,, \qquad \text{for } j \in \mfM
\\
\mSigma &= \frac{1}{N} \sum_{n=1}^{N} \sum_{k=1}^{m} a(\Theta_{nk}) \left( \vx_n - \vgamma^{(k)} \right) \left( \vx_n - \vgamma^{(k)} \right)^\intercal\,,
\\
\mZ_{jk} &= \frac{\sum_{n=1}^{N} b(\Theta_{n-1,j}, \Theta_{nk})}{\sum_{n=1}^{N} \sum_{k=1}^{m} b(\Theta_{n-1,j}, \Theta_{nk})} \, \qquad \text{for } j,k \in \mfM
\\
\mathfrak{p}_k &= \frac{a(\Theta_{1k})}{\sum_{k=1}^{m} a(\Theta_{1k})} \, \qquad \text{for } k \in \mfM
\end{align*}
This is a simple extension of the usual Gaussian case to the shared covariance case.

\clearpage 
\end{appendices}

\bibliographystyle{chicago}
\bibliography{OutperformanceAndTracking_partialInfo}

\end{document}